\documentclass[12pt,a4paper,reqno]{amsart}

\usepackage[margin=2.8cm,footskip=1cm]{geometry}
\usepackage[utf8]{inputenc}
\usepackage[T1]{fontenc}
\usepackage[english]{babel}
\usepackage{amsmath}
\usepackage{amsfonts}
\usepackage{amssymb}
\usepackage{amsthm}
\usepackage{ucs}
\usepackage{graphicx}
\usepackage{xcolor}
\usepackage{dsfont}
\usepackage{tikz}
\usepackage{wasysym}
\usepackage{physics}
\usepackage{mathtools}
\usepackage{xifthen}
\usepackage[textwidth=2.5cm,textsize=tiny]{todonotes}
\usepackage{enumitem}
\usepackage{stmaryrd}
\usepackage{constants}
\usepackage[ruled]{algorithm2e}
\usepackage{tikz,pgfplots}
\usepackage[hidelinks]{hyperref}
\usepackage[font={small}]{caption}

\usepackage{constants}

\newconstantfamily{ctrlcst}{
	symbol=\kappa
}
\newconstantfamily{csts}{
	symbol=C
}
\renewconstantfamily{normal}{
	symbol=c
}

\makeatletter
\newcommand{\pushright}[1]{\ifmeasuring@#1\else\omit\hfill$\displaystyle#1$\fi\ignorespaces}
\newcommand{\pushleft}[1]{\ifmeasuring@#1\else\omit$\displaystyle#1$\hfill\fi\ignorespaces}
\makeatother

\renewcommand{\norm}[1]{\|#1\|}
\newcommand{\normI}[1]{\left\|#1\right\|_{\scriptscriptstyle 1}}
\newcommand{\normsup}[1]{\left\|#1\right\|_{\scriptscriptstyle\infty}}

\newcommand{\supp}{\mathrm{supp}}

\newcommand{\lrangle}[1]{\langle #1 \rangle}

\newcommand{\Id}{\mathrm{Id}}

\newcommand{\Z}{\mathbb{Z}}
\newcommand{\R}{\mathbb{R}}
\newcommand{\Rd}{\mathbb{R}^d}
\newcommand{\Zd}{\mathbb{Z}^d}

\newcommand{\Ed}{\mathbb{E}^d}

\newcommand{\pairings}{\mathrm{Pairings}}

\newcommand{\bbB}{\mathbb{B}}
\newcommand{\bbE}{\mathbb{E}}

\newcommand{\bbR}{\mathbb{R}}
\newcommand{\bbS}{\mathbb{S}}
\newcommand{\bbZ}{\mathbb{Z}}

\newcommand{\rme}{\mathrm{e}}
\newcommand{\rmB}{\mathrm{B}}
\newcommand{\rmJ}{\mathrm{J}}

\newcommand{\calE}{\mathcal{E}}
\newcommand{\calF}{\mathcal{F}}
\newcommand{\calG}{\mathcal{G}}
\newcommand{\calH}{\mathcal{H}}
\newcommand{\calP}{\mathcal{P}}

\theoremstyle{plain}
\newtheorem{theorem}{Theorem}[section]
\newtheorem{lemma}[theorem]{Lemma}
\newtheorem{proposition}[theorem]{Proposition}
\newtheorem{corollary}[theorem]{Corollary}

\newtheorem{remark}{Remark}[section]

\newtheorem{claim}{Claim}

\theoremstyle{definition}
\newtheorem{definition}{Definition}[section]
\newtheorem{obs}{Observation}

\author{Alessandro Giuliani}
\address{Dipartimento di Matematica e Fisica, Universit\`a degli Studi Roma Tre, L.go S. L. Murialdo 1, 00146 Roma, Italy \& \\
Centro Linceo Interdisciplinare {\it Beniamino Segre}, Accademia Nazionale dei Lincei, Via della Lungara 10, 00165 Roma, Italy}
\email{alessandro.giuliani@uniroma3.it}

\author{S\'{e}bastien Ott}
\address{D\'epartement de Math\'ematiques, Universit\'e de Fribourg,
Chemin du Mus\'ee 23, 1700 Fribourg, Switzerland}
\email{ott.sebast@gmail.com}

\date{\today}

\title{Low temperature asymptotic expansion for classical $O(N)$ vector models}

\begin{document}

\begin{abstract}
We consider classical $O(N)$ vector models in dimension three and higher and investigate the nature of the low-temperature expansions for their multipoint 
spin correlations. We prove that such expansions define asymptotic series, and derive explicit estimates on the error terms associated with their finite 
order truncations. The result applies, in particular, 
to the spontaneous magnetization of the 3D Heisenberg model. The proof combines a priori bounds on the moments of the local spin observables, following from 
reflection positivity and the infrared bound, with an integration-by-parts method applied systematically 
to a suitable integral representation of the correlation functions. Our method generalizes an  approach, proposed originally by Bricmont and collaborators \cite{Bricmont+Fontaine+Lebowitz+Lieb+Spencer-1980} in the context of the rotator model, to the case of non-abelian symmetry and non-gradient observables. 
\end{abstract}

\maketitle


\section{Introduction}

In 1976, in one of the most influential papers in mathematical statistical mechanics, Fr\"ohlich-Simon-Spencer (FSS)  \cite{Frohlich+Simon+Spencer-1976} proved the existence
of orientational long range order for classical $O(N)$ vector models in three or more dimensions. In 1981, in the special case of the $O(2)$ model, also known as the rotator or XY model,
Bricmont-Fontaine-Lebowitz-Lieb-Spencer (BFLLS) 
\cite{Bricmont+Fontaine+Lebowitz+Lieb+Spencer-1980} extended the 
FSS result, proving, in particular, that the formal low-temperature expansion for its magnetization defines an asymptotic series; that is, the difference between 
the magnetization at temperature $T$ and the truncation of its low-temperature series at any finite order $n$ is $o(T^n)$ as $T\to 0$. A direct generalization of this result to $O(N)$ models with $N\ge 3$, including the 
physically relevant case of the Heisenberg model, remained open since, due to the difficulty in extending the BFLLS method to the case of non-abelian rotational symmetry. 
By `direct',  here, we mean by similar elementary methods, combining  a priori bounds on the moments of local spin observables, based on reflection positivity and the related infrared bound \cite{Frohlich+Simon+Spencer-1976, Frohlich+Israel+Lieb+Simon-1978}, with systematic integration-by-parts applied to a suitable integral representation of
the spin correlations. In this paper we provide such a direct, self-contained, proof of the asymptotic nature of the low-temperature expansion for the magnetization and other 
spin correlation functions of classical $O(N)$ vector models, for any $N\ge 3$. Technically, our proof is based on an extension of the BFLLS proof, capable of handling non-abelian symmetry and non-gradient observables. 

Let us recall that there are other, more involved, approaches that can  be used to investigate the low-temperature properties of classical $O(N)$ vector models: we refer here to Balaban's multiscale analysis \cite{Balaban-1995, Balaban-1996, Balaban-1998} of its correlation functions \cite{Balaban+OCarroll-1999}. Balaban's construction implies, in  particular, 
that the low-temperature expansion for the magnetization is an asymptotic series; it actually provides much more detailed informations on  the correlations at low temperatures, including their 
large distance asymptotics and decay exponents. However, the construction is extremely involved and it is hard to extract from it explicit informations in a simple and direct way. 

For this reason, it is important to develop simpler approaches to the problem, particularly in cases, such as the classical $O(N)$ vector model, where the system displays reflection positivity, 
which in turn implies several strong a priori bounds on correlation functions. 
The proof described below shows that, even in models with non-abelian continuous symmetry, 
there is no intrinsic obstacle in using the infrared bound to derive bounds on the low-temperature expansion for correlations. 
However, even if our method allows us to control multipoint spin correlation at any fixed coordinates, 
it does not provide informations on their asymptotic behavior as the relative distance between the spins diverges to infinity. We believe that, in order to compute the 
critical exponents and the dominant asymptotics of spin correlations, multiscale analysis is 
inevitable. It would be extremely nice to be able to use systematically reflection positivity within a multiscale scheme, thus simplifying Balaban's approach, particularly in 
connection with 
the technically involved procedure related to the small/large field decomposition. It would also be nice to develop methods, inspired by the BFLLS approach, capable of 
controlling the low-temperature series for the magnetization of reflection positive quantum spin systems, such as the 3D quantum XY model or quantum Hesienberg anti-ferromagnet. 
These long term goals are behind the very motivations of the current work. It remains to be seen whether such challenging problems can be solved or at least attacked via
extensions of the current work, and we hope to come back to these problems in future publications. 

\subsection{Definition of the model and finite-volume infrared bound.}\label{sec:1.1}

Let \(\Lambda_L\) be the \(d\)-dimensional discrete torus of side \(L\), in dimension $d\ge 3$. Each site of $\Lambda_L$ carries an $N$-component unit vector, called `spin'; we assume that $N\ge 2$, and denote 
by \(\Omega_L = (\bbS^{N-1})^{\Lambda_L} \) the corresponding space of spin configurations. For \(f:\Omega_L \to \R\), \(\beta\geq 0\), and \(h\geq 0\), we define the un-normalized expectations as follows:
\begin{equation}\label{Zetas}	\begin{split}
&	Z_{L;\beta}(f) = \int_{\Omega_L} \prod_{x\in\Lambda_L} d\nu_N(S_x) f(S) e^{\beta \sum_{x\in\Lambda_L}\sum_{e} S_x\cdot S_{x+e} },\\
&	Z_{L;\beta, h}(f) \equiv Z_{L;\beta}\big(f  e^{h\sum_{x\in\Lambda_L} S_x^N }\big),
\end{split}\end{equation}
where \(\nu_N\) is the Lebesgue (area) measure on \(\bbS^{N-1}\). 
Their normalized counterparts are:
\begin{equation}\label{eq:averages}
	\lrangle{f}_{L;\beta} = \frac{Z_{L;\beta}(f)}{Z_{L;\beta}(1)},\quad \lrangle{f}_{L;\beta,h} = \frac{Z_{L;\beta,h}(f)}{Z_{L;\beta,h}(1)}
\end{equation}We will write \(Z_{L;\beta,h}\equiv Z_{L;\beta,h}(1)\), and often omit \(\beta\) from the notation. \(\mu_{L;h}\) will denote the probability measure with expectation \(\lrangle{\ }_{L;h}\). 

{\it Reflection positivity}  of $\mu_{L;h}$ implies the following finite-volume version of the \emph{infrared bound} ~\cite[Theorem 2.1]{Frohlich+Simon+Spencer-1976}, 
which is one of the basic ingredients of our analysis. 
We denote by $\hat\rme_k$ with $k=1,\ldots,d$ the elements of the canonical basis of unit vectors of $\mathbb R^d$, and 
$\nabla^{\hat\rme_k}_x f=  f(x+\hat\rme_k)-f(x) = (\nabla^{\hat \rme_k}f)(x)$. We also denote by \(\Delta_{L}\) the Laplacian in the torus \(\Lambda_L\) and by \(G_L = (-\Delta_L)^{-1}\)
its finite-volume Green's function.

\begin{proposition}
	\label{thm:Infrared_finite_volume.0}
	For any \(h,\beta\geq 0\), any \(L\), and any \(f_1,\cdots,f_d :\Lambda_L\to\R^{N}\),
	\begin{equation*}
	\lrangle{e^{\sum_{k=1}^d\sum_{x\in\Lambda_L} S_x\cdot\nabla^{\hat\rme_k}_xf_k }}_{L;h}\leq e^{\frac{1}{2\beta}\sum_{k=1}^d\sum_{x\in\Lambda_L}(f_k(x))^2 }.
	\end{equation*} In particular, 
	for \(f\) such that \(\sum_{x}f(x) =0\), letting $f_k(x)=(-\Delta_L)^{-1}\nabla^{\hat\rme_k}_x f$ in  the previous equation, one has:
	\begin{equation}
	\label{eq:infrared_fin_vol.0}
	\lrangle{e^{\sum_{x\in\Lambda_L} S_x\cdot f(x) }}_{L;h}\leq e^{\frac{1}{2\beta}(f,G_L f) }.
	\end{equation}
\end{proposition}
 
\subsection{Construction of the infinite volume measure and infrared bound.}\label{sec:1.1bis}

We will work with the successive limits \(L\to \infty\) (thermodynamic limit), and \(h\searrow 0\) of \(\mu_{L;h}\). We denote by $\mu_h$ a (arbitrarily chosen) 
cluster point of $\mu_{L;h}$ as $L\to\infty$, which exists by compactness and inherits translation invariance from the finite volume measures, and by $\langle\cdot\rangle_h$ 
the corresponding expectation. It is believed (but, to the best of the authors knowledge, not yet proved) that the limit of $\mu_{L;h}$ as \(L\to \infty\) always exists and is exponentially 
mixing when \(h>0\). 

A crucial ingredient for our construction is the following infinite volume version of the infrared bound: 
\begin{equation}\lrangle{e^{\sum_{x} (S_x-\lrangle{S_x}_h)\cdot f(x) }}_h\leq e^{\frac{1}{2\beta}(f,G f) },\label{eq:ch}\end{equation}
for any $f$ of finite support, with $G=(-\Delta)^{-1}$ the Green's function associated with the lattice Laplacian \(\Delta\) on \(\Zd\). As proved in \cite[Corollary 2.5]{Frohlich+Simon+Spencer-1976}, if $\mu_h$ is {\it ergodic}, then \eqref{eq:ch} readily follows from \eqref{eq:infrared_fin_vol.0}. 
When \(N=2,3\), one can extract from Lee-Yang theory (see~\cite{Frohlich+Rodriguez-2012,Frohlich+Rodriguez-2017}) that the limit state \(\mu_h=\lim_{L\to \infty} \mu_{L;h}\) (with \(h>0\)) is well defined and exponentially mixing; in particular, it is ergodic and, therefore, \eqref{eq:ch} follows. When \(N>3\), Lee-Yang theorem is not available. However, the day can still be saved. In fact, for the purpose of proving \eqref{eq:ch}, ergodicity is not required: as shown in Appendix \ref{prova:IB}, 
 the property that $L^{-d}\sum_{n\in \mathbb Z^d: |n_i|\le L/2}S_{x+n}$ converges to $\lrangle{S_x}_h$ in $L^p(\mu_h)$ for $p\ge 1$ suffices. 
 
Let \(\psi:\R^N\to \R\) be the pressure of the \(O(N)\) model, defined by \(\psi(v)\equiv \psi_{\beta}(v) = \lim_{L\to\infty} \frac{1}{|\Lambda_L|} \log Z_{L;\beta}(e^{\sum_{x\in \Lambda_L} v\cdot S_x})\). Let $$D^*:=\{h>0:\  \psi \textnormal{ is differentiable at } hs\ \forall s\in\mathbb S^{N-1}\}$$ 
(note that the definition is well-posed, thanks to $O(N)$ invariance: if $\psi$ is differentaible at $hs_0$ for some $s_0\in\mathbb S^{N-1}$, then it is automatically differentiable at $hs$ for all $s\in\mathbb S^{N-1}$). $D^*$ has full Lebesgue measure on $(0,+\infty)$, thanks to the convexity of $\psi$. Moreover, let $\mu$ be a (arbitrarily chosen) limit point of $\mu_h$, as $h\in D^*$ tends to $0^+$, and let $\lrangle{\cdot}$ be the corresponding expectation. Then the following generalization of \cite[Corollary 2.5 and Theorem 3.1]{Frohlich+Simon+Spencer-1976} holds (for the proof, see Appendix \ref{prova:IB}).

\begin{proposition}
	\label{thm:Infrared_infinite_volume_magn_bound} 	
If $h\in D^*$, for any \(f\) with finite support,
	\begin{equation}
	\label{eq:infrared_infin_vol_centered_spins}
	\lrangle{e^{\sum_{x} (S_x-\lrangle{S_x}_h)\cdot f(x) }}_h\leq e^{\frac{1}{2\beta}(f,G f) }.
	\end{equation}
In particular,
\begin{equation}
	\label{eq:infrared_magnet}
	\lrangle{S_0}_h\cdot\lrangle{S_0}_h \geq 1-{N T}G(0,0),
	\end{equation}where \(G(0,0) = (2\pi)^{-d}\int_{p\in[-\pi,\pi]^d} dp \big(2\sum_{k=1}^d (1-\cos(p_k))\big)^{-1}\).
	The bounds \eqref{eq:infrared_infin_vol_centered_spins} and \eqref{eq:infrared_magnet} also hold with $\lrangle{\cdot}$ replacing $\lrangle{\cdot}_h$.
\end{proposition}

\subsection{Main results}

We are now ready to state the main result of this work.  

\begin{theorem}
	\label{thm:main}
	For any \(n,K\in \Z_+:=\{0,1,2,\cdots\}\), and \(\epsilon>0\) there exists \(C<\infty\) such that for any \(\alpha^1,\cdots,\alpha^N:\Zd\to \Z_+\) with \(\sum_{k=1}^N\sum_{x\in \Zd}\alpha^k_x \leq K\), there are coefficients \(a_0,\cdots , a_n\in \R\) such that
	\begin{equation*}
		\Big|\lrangle{\prod_{x}\prod_k (S_x^k)^{\alpha_{x}^k}} -\sum_{i=0}^{n} a_i T^i\Big| \leq CT^{n+1-\epsilon}
	\end{equation*}where the \(a_i\)s are given by an explicit inductive algorithm.
\end{theorem}

A Corollary of Theorem~\ref{thm:main} and of our inductive algorithm is the following second order approximation of the magnetization. Recall that 
\(\Delta\) denotes the Laplacian on \(\Zd\) and \(G=(-\Delta)^{-1}\) its Green function. Recall also that $\{\hat\rme_k\}_{k=1,\ldots,d}$ denotes the canonical basis of unit vectors of $\mathbb R^d$. The discrete derivative of $G$ acting (say) on its second coordinate, $\nabla^{\hat\rme_k}_xG(x_0,\cdot)$, was defined right before the statement of Proposition \ref{thm:Infrared_finite_volume.0}. Similarly, we let 
$\nabla^{\hat\rme_k,\hat\rme_{k'}}_{x,y} G = G(x+\hat\rme_k,y+\hat\rme_{k'})-G(x,y+\hat\rme_{k'})-G(x+\hat\rme_k,y)+G(x,y)$. 

\begin{corollary}
	\label{cor:magnet_second_order} For $d\ge 3$ and $N\ge 2$, the following second order expansion for the spontaneous magnetization of the $O(N)$ vector model holds, as $T\to 0^+$: 
	\begin{multline}\label{eq:1}
		\lrangle{S_0^N} = 1 -T \frac{(N-1)}{2}  G(0,0)
		+ T^2 \frac{(N-1)}{2} \Big[ \frac{(3N-5)}4G(0,0)^2\\
		- \sum_{x_1\in \mathbb Z^d}\sum_{e\in \{\hat\rme_i, i=1,\cdots d\}}\big(\nabla_{x_1}^{e} G(0,\cdot)\big)^2\Big(\frac{1}{2} \nabla_{x_1,x_1}^{e,e}G + (N-2)G(0,0)\Big)\Big] + o(T^2).
	\end{multline}
\end{corollary}

Of course, similar formulas can be derived for third or higher order truncations, but we decided to spell out for illustrative purposes only the simplest non-trivial one. 

A second immediate consequence of the proof of our main result is the following. 

\begin{corollary}\label{cor:2}
Let \(S=(S_x)_{x\in\mathbb Z^d}\) be distributed with respect to the measure $\mu$. Then the first $N-1$ components of the rescaled field $\sqrt{\beta}S$, that is, \(\sqrt{\beta}(S^1,S^2,\dots,S^{N-1})\), converge in law, as \(\beta\to \infty\), to \(N-1\) independent Gaussian Free Fields on $\mathbb Z^d$.
\end{corollary}

\subsection{Organization of the paper}

We will prove Theorem~\ref{thm:main} by introducing an inductive procedure to compute the coefficients in the Taylor series. This expansion follows the same scheme as in~\cite{Bricmont+Fontaine+Lebowitz+Lieb+Spencer-1980}: a regularized version of the Gaussian integration by part and a priori bounds on the moments of suitable observables, obtained by using the infrared bound. The main difference lies in the fact that for \(N>2\), the \(O(N)\) model can not be written as a gradient perturbation of a Gaussian free field. This substantially complicates the picture and, as a result, we need to introduce a more involved procedure than the one of~\cite{Bricmont+Fontaine+Lebowitz+Lieb+Spencer-1980}, 
particularly in regards to the procedure required for summing over the coordinates the remainder terms. The 
paper is organized as follows: after having summarized the notations used in the paper in Section \ref{subsec:notation}, in Section \ref{sec:2} we introduce the coordinates we use to parametrize the spin space; then, using these coordinates, we express the Gibbs measure 
as a perturbation of a Gaussian one; finally, we use this re-writing of the measure to derive the formal low-temperature expansion, whose structure will be used below, in the 
description of the inductive procedure. In Section \ref{sec:infraref_bnd}, we explain how to use the infrared bound to derive a priori bound on the spin moments (this is 
an adaptation of the method in \cite{Bricmont+Fontaine+Lebowitz+Lieb+Spencer-1980}). 
In Section \ref{sec:int_by_part} we describe the version of the Gaussian integration by parts to be used in the following. Finally, in Section \ref{sec:Inductive_exp}, 
we describe the inductive procedure: we first illustrate the procedure for the computation of the spontaneous magnetization 
at order $T^2$, thus proving Corollary \ref{cor:magnet_second_order}, and then discuss the general scheme at all orders, thus proving Theorem \ref{thm:main} and Corollary \ref{cor:2}. 
Appendix \ref{prova:IB} contains the proof of Proposition \ref{thm:Infrared_infinite_volume_magn_bound}. 
The proofs of the priori bounds on the decay of correlations (which use reflection positivity in a 
relatively standard manner) are collected in Appendix \ref{app:UB_corr_funct}, while Appendix \ref{app:B} 
discusses some standard Gaussian estimates, and their implications for the 
summation over the coordinates of certain connected Gaussian expectations, which systematically appear in the evaluation of the remainder terms. 

\section{Notations}\label{subsec:notation}

Before starting the proof of our main results, let us summarize here the notations and conventions used in this paper. For ease of reference, we include here also those already introduced in the previous section.

We write \(\Z_+=\{0,1,2,\cdots\}\) and \([n] = \{1,2,\cdots, n\}\).

We will work in dimensions \(d\geq 3\), on either \(\Zd\) or \(\Lambda_L=\{0,\cdots,L-1\}^d\) the \(d\)-dimensional discrete torus of side \(L\). We denote \(E_L,E\) the set of nearest neighbour pairs in \(\Lambda_L,\Zd\). We write \(i\sim j\) when \(i\) and \(j\) are nearest neighbours. When writing \(\sum_{x}\), we mean summing over $x\in\mathbb Z^d$. 

The  symbols \(\hat\rme_i\), with \(i\in\{1,\ldots,d\}\), will denote the canonical basis unit vectors of \(\Rd\) and \(\rmB = \{-\hat\rme_1,\hat\rme_1,\ldots,-\hat\rme_d,\hat\rme_d\}\). Sometimes we will also use the notation \(\hat\rme_i\), with $i\in\{1,\ldots,N\}$, for the base vectors of \(\bbR^N\) and hope that it will not confuse the reader. We attribute a canonical orientation to each edge: any edge can be written \(\{x,x+e\}\) with \(e \in \rmB_+=\{\hat\rme_i, i=1,\ldots, d\}\), and we orient it from \(x\) to \(x+e\). When writing \(\sum_{e}\), we mean summing over \(e\in \rmB_+\). Sums over \(e\in \rmB\) will be explicitly mentioned. Let \(\bbE_d=\big\{(x,e):\ x\in\Zd,e\in \rmB_+\big\}\). There is a canonical bijection between \(E\) and \(\bbE_d\).

For \(p:\Zd\to\Z_+\), \(q:\bbE_d\to \Z_+\), and \(\tilde{p}:\Zd\times[N-2]\to \Z_+\) we write
\begin{equation}\label{gat}\begin{split}
&	\phi^p = \prod_{x\in\Zd}\phi_x^{p_x},\quad (\nabla\phi)^{q} = \prod_{(x,e)\in\bbE_d} (\nabla_x^{e}\phi)^{q_{x,e}},\quad u^{\tilde{p}} = \prod_{x}\prod_{k=1}^{N-2} (u_x^k)^{\tilde{p}_x^k},\\
&	\normI{p}= \sum_{x\in \Zd} |p_x|,\quad \normI{q} = \sum_{(x,e)\in \bbE_d} |q_{x,e}|,\quad \normI{\tilde{p}} = \sum_{x}\sum_{k} |\tilde{p}_x^k|,
\end{split}\end{equation}
where: \(\tilde{p}^k:\Zd\to\Z_+\) is given by \(\tilde{p}^k_x= \tilde{p}(x,k)\), and, for \(f:\Zd\to\R^{N}\),
\begin{equation*}
	\nabla^{e}_x f=  f(x+e)-f(x) = (\nabla^{e}f)(x) ,\ e\in \rmB.
\end{equation*}
We say that  \(p\) is odd if \(\normI{p}\) is, same for \(q\). We say that \(\tilde{p}\) is even if each \(\tilde{p}^k\) is, we say that it is odd if there is \(k\in\{1,\cdots,N-2\}\) such that \(\tilde{p}^k\) is. We define the \emph{support} of \(p,q,\tilde{p}\) as above:
\begin{equation}\label{supp}\begin{split}
&	\supp_p = \{x\in\Zd:\ p_x\neq 0 \},\quad \supp_{\tilde{p}} = \bigcup_{k=1}^{N-2} \supp_{\tilde{p}^k},\\
&	\supp_q = \bigcup_{(x,e)\in\bbE_d: q_{x,e}\neq 0 } \{x,x+e\}.
\end{split}\end{equation}
We also let $\mathfrak P,\mathfrak Q,\tilde{\mathfrak P}$ be the sets of tuples \(p,q,\tilde{p}\) of finite support of the form described above, respectively. 
Analogous definitions hold for tuples $p:\Lambda_L\to\Z_+$, \(q:\big\{(x,e):\ x\in\Lambda_L, e\in \rmB_+\big\}\to \Z_+\), and \(\tilde{p}:\Lambda_L\times[N-2]\to \Z_+\), 
whose sets will be denoted $\mathfrak P_L,\mathfrak Q_L,\tilde{\mathfrak P}_L$, respectively.

\medskip

For \(f:\Zd\to\R^{N}\) and \(g:\Z^d\times\Zd\to \R\), let
\begin{gather*}
(-\Delta f)(x) = 2d f(x) - \sum_{e\in\rmB} f(x+e) =  -\sum_{e\in\rmB}\nabla^{e}_x f,\\
\nabla^{e,e'}_{x,y} g = g(x+e,y+e')-g(x,y+e')-g(x+e,y)+g(x,y), \quad e,e'\in\rmB.
\end{gather*}

\(\Delta\) is the Laplacian on \(\Zd\), we let \(G=(-\Delta)^{-1}\) be its Green function. \(G^{m}\) will denote the Green function of the massive Laplacian: \(G^{m} = (-\Delta + m^2\Id)^{-1}\). We write \(\Phi_m\) for the law (and expectation) of the (massive) Gaussian Free Field on \(\Zd\) (the centred Gaussian field with covariance \(G^m\)). \(\Delta_{L}\) will denote the Laplacian in the torus \(\Lambda_L\) and \(G^{m}_L = (-\Delta_L + m^2\Id)^{-1}\). \(\Phi_{m;L}\) will denote the law (and expectation) of the centred Gaussian vector with covariance matrix \(G^{m}_L\). When \(m\) is omitted, it is set to \(0\). We will write \(G^m(\cdot, y)\) the function \(x\mapsto G^{m}(x,y)\), and define similarly \(G^{m}(y,\cdot)\). We will sometimes use the notation \(G_{xy}^m \equiv G^m(x,y)\).

Scalar product will be denoted \(\cdot\) or \((\ ,\ )\). By convention, if \(f,g\) have value in \(\R^N\), \((f,g)=\sum_{x\in\Zd} f(x)\cdot g(x)\).

We will  denote \(\bbS^{k}\) the unit sphere in \(\bbR^{k+1}\) and \(\nu_{k+1}\) the Lebesgue (un-normalized uniform) measure on \(\bbS^{k}\). We also denote \(\bbB^k=\{x\in\R^k:\ \norm{x}\leq 1 \}\). Moreover, $\Omega_L=(\bbS^{N-1})^{\Lambda_L}$.

Finally, we recall that $\mu_{L;\beta,h}$, or simply $\mu_{L;h}$ (we often omit \(\beta\) from the notation), denotes the probability measure associated with the 
average $\lrangle{\cdot}_{L;\beta,h}$, or simply $\lrangle{\cdot}_{L;\beta}$, in \eqref{eq:averages}. $Z_{L;\beta,h}\equiv Z_{L;\beta,h}(1)$, or simply $Z_{L;h}$ denotes the partition function (see \eqref{Zetas}), and \(\psi:\R^N\to \R\) the pressure in the thermodynamic limit, defined as 
$\psi(v)=\lim_{L\to\infty}|\Lambda_L|^{-1}\log Z_{L;\beta}(e^{\sum_{x\in \Lambda_L} v\cdot S_x})$.
$\mu_h$ denotes an arbitrarily chosen limit point of $\mu_{L;h}$ as $L\to\infty$, and 
$\lrangle{\cdot}_h$ denotes the corresponding average. $\mu$ denotes an arbitrarily chosen limit point of $\mu_{h}$ as $h\searrow 0$ taken along a sequence of $h$s in $D^*:=\{h>0:\   \psi \textnormal{ is differentiable at } hs\ \forall s\in \bbS^{N-1},\}$, and $\lrangle{\cdot}$ the corresponding average.

In the following sections, $C,C',\ldots, c,c',\ldots$ denote constants independent of $T,h,m$, whose specific values may change from line to line; dependence upon additionally auxiliary parameters will usually be specified explicitly, with the exception of $N\ge 2$, which we shall not track the dependence of the various constants on. We will explicitly focus on the case $N>2$, the case $N=2$ being significantly simpler, in that the $u$ variables introduced at the beginning of the next section are not present; the interpretation and adaptation of the proof   
to the case $N=2$ is immediate and left to the reader (moreover, the case $N=2$ is already covered by \cite{Bricmont+Fontaine+Lebowitz+Lieb+Spencer-1980}). 

\section{Gaussian coordinates and the formal low-temperature expansion}\label{sec:2}

\subsection{The Gibbs measure as a perturbed Gaussian}
\label{sec:perturbed_Gaussian_pic}

We will use a particular parametrization of \(\bbS^{N-1}\): for \(u\in\R^{N-2} \) with \(\norm{u}\leq 1\), and \(-\pi< \theta\leq \pi\), we let 
\begin{equation}\label{transf}
	\Pi(u,\theta) = \big( u_1,\cdots, u_{N-2}, \rho\sin(\theta), \rho\cos(\theta) \big)
\end{equation}
with \(\rho = \rho(u) = \sqrt{1-\norm{u}^2}\), which provides a parametrization of \(\bbS^{N-1}\). Note that if $N=2$ there is no variable $u$ and $\rho$ should be replaced by $1$ (in this case, many of the formulas and equations below simplify significantly). One has the change of coordinate formula: for any integrable 
\(f:\bbS^{N-1}\to\R\),
\begin{equation}\label{eq:3bis}
	\int_{\bbS^{N-1}} d\nu_N(S) f(S) = \int_{\bbB^{N-2}} du\, \int_{-\pi}^{\pi} d\theta f(u_1,\cdots,u_{N-2},\rho \sin\theta,\rho\cos\theta).
\end{equation}
In light of \eqref{eq:3bis}, given an observable $a:\Omega_L\to \R$ of the spin configuration in $\Lambda_L$, 
letting $A$ be the composition of $a$ with the transformation \eqref{transf}, to be applied to each spin $S_x\in \bbS^{N-1}$, $x\in\Lambda_L$, we find that the un-normalized 
expectation of $a$ can be rewritten as follows: 
\begin{equation}\label{eq:3}
\begin{split}
 Z_{L;\beta,h}(a) & = \int_{\Omega_L} \prod_{x\in \Lambda_L} d\nu_N(S_x)\, a(S) e^{\sum_{x\in\Lambda_L}(\beta\sum_{\rme} S_{x+\rme}\cdot S_x+h S^N_x)}\\
 &= \int\limits_{(\bbB^{N-2})^{\Lambda_L}}\! d{u}\! \int\limits_{[-\pi,\pi]^{\Lambda_L}}\!\!\!d\theta \, A(u,\theta) e^{H_{L;h}(u,\theta)} ,
\end{split}
\end{equation}
where, in the first line, $S=(S_x)_{x\in\Lambda_L}$, and 
\begin{equation} 
\begin{split} H_{L;h}(u,\theta) &= \sum_{x\in\Lambda_L}\Big[
\beta \sum_{e}\Big(u_x\cdot u_{x+\rme}+\sqrt{(1-\|u_x\|^2)(1-\|u_{x+e}\|^2)}\cos(\nabla^e_x\theta)\Big)\\
& + h \sqrt{(1-\|u_x\|^2)}\cos(\theta_x)\Big].\end{split}
\end{equation}
If we now rescale the variables as follows: 
\begin{equation*}
	\phi_x=\sqrt{\beta}\theta_x,\quad \tilde{u}_x = \sqrt{\beta}u_x, \label{eq:rescale}
\end{equation*}
the expectation in \eqref{eq:3}  becomes, letting $T=\beta^{-1}$ and $\tilde A(\tilde u,\phi)= A(\sqrt{T}\tilde u,\sqrt{T}\phi)$:
\begin{equation}\label{eq:5}
\begin{split}
 Z_{L;\beta,h}(a) & = T^{L^d(N-1)/2}\!\!\!\!\int\limits_{(\sqrt{\beta}\bbB^{N-2})^{\Lambda_L}}\!\!\! \!d{\tilde u}\, e^{\mathcal H_{L;h}(\tilde u)}\!\!\!\!
  \int\limits_{[-\sqrt{\beta}\pi,\sqrt{\beta}\pi]^{\Lambda_L}}\!\!\!\!d\phi \, \tilde A(\tilde u,\phi) e^{-\frac12\|\nabla\phi\|^2}
  e^{W_{L;h}(\tilde u,\phi)}
\end{split}
\end{equation}
where $\|\nabla \phi\|^2=\sum_{x\in\Lambda_L}\sum_{e}(\nabla^e_x\phi)^2$, $\|\phi\|^2=\sum_{x\in\Lambda_L}(\phi_x)^2$,  
\begin{equation}
\label{eq:def:W_T}
\begin{split}  \mathcal H_{L;h}(\tilde u)&= \sum_{x\in\Lambda_L} \Big[\sum_{e}\Big(\tilde u_x\cdot \tilde u_{x+e}+\beta \sqrt{(1-T\|\tilde u_x\|^2)(1-T\|\tilde u_{x+e}\|^2)}\Big) \\
& +h \sqrt{1-T\|\tilde u_x\|^2}\Big]
 \end{split}\end{equation}
 and
 \begin{equation}\label{eq:7}\begin{split}
W_{L;h}(\tilde u,\phi)&= \sum_{x\in\Lambda_L}\Big\{\sum_{e}\Big[\beta \sqrt{(1-T\|\tilde u_x\|^2)(1-T\|\tilde u_{x+e}\|^2)}\Big(\cos(\sqrt{T}\nabla^e_x\phi)-1\Big)\\ &+\frac12(\nabla^e_x\phi)^2\Big] + h\sqrt{1-T\|\tilde u_x\|^2}\big(\cos(\sqrt{T}\phi_x)-1\big)
\Big\}.\end{split}
\end{equation}
Note that $W_{L;h}$ vanishes linearly in $T$ as $T\to 0$. More precisely, by expanding in Taylor series in $T$ the right side of \eqref{eq:7}, we find that
$W_{L;h}$ admits the following low-temperature expansion: 
\begin{equation}\label{eq:8}W_{L;h}(\tilde u, \phi) = \sum_{s\ge 1}T^s \sum_{x\in\Lambda_L}\sum_{\substack{r,p\ge 0: \\ r+p=s}}\Big[ c_r \sum_{e}
(\nabla^e_x\phi)^{2r+2} \mathcal G^p_{x,e}(\tilde u)+c_{r-1}\mathds 1_{r\ge 1}h (\phi_x)^{2r} \tilde c_p\|\tilde u_x\|^{2p}\Big],\end{equation}
where  $c_r=\tfrac{(-1)^{r+1}}{(2r+2)!}$, $\tilde c_p=(-1)^p{1/2 \choose p}$ (recall that \(\binom{1/2}{0} = 1\) and \(\binom{1/2}{p} = \frac{1}{p!} \frac{1}{2}(\frac{1}{2}-1)\cdots (\frac{1}{2}-p+1)\) for integer \(p>0\)) and $\mathcal G^p_{x,e}(\tilde u)=\sum_{l+m=p}\tilde c_l\tilde c_m \|\tilde u_x\|^{2l}\|\tilde u_{x+e}\|^{2m}$.

Therefore, \eqref{eq:5} expresses the \(O(N)\) model as a perturbation of a Gaussian measure, as far as the integration over the $\phi$ variables is concerned. Moreover, 
as a consequence of the \(O(N-1)\) symmetry of the model in the hyper-plane orthogonal to the magnetic field, for any \(k\in\{1,\cdots,N-2\}\) one has the equality in law
(with respect to the probability measure with density $\propto  e^{\mathcal H_{L;h}(\tilde u)-\frac12\|\nabla\phi\|^2
+W_{L;h}(\tilde u,\phi)}$):
	\begin{equation}
		\label{eq:remaining_symmetry_u_theta}
		\big(\tilde{u}^1,\cdots, \tilde{u}^{N-2}\big) \stackrel{\textnormal{law}}{=} \Big(\tilde{u}^1, \cdots, \tilde{u}^{k-1}, \mathfrak F(\tilde u,\phi), \tilde{u}^{k+1}, \cdots,\tilde{u}^{N-2}\Big).
	\end{equation}
where $\mathfrak F(\tilde u,\phi)=\big(\sqrt{\beta(1-T\|\tilde u_x\|^2)}\sin \big(\sqrt{T}\phi_x\big)\big)_{x\in\Lambda_L}$,

\subsection{The formal low temperature expansion} 
\label{sec:formal_low_temp_exp}

Eqs.\eqref{eq:5},\eqref{eq:8},\eqref{eq:remaining_symmetry_u_theta} allow us to recursively define the coefficients of the (a priori formal) low-temperature expansion for the average 
of any observable $a:\Omega_L\to\mathbb R$. Let $\tilde A(\tilde u,\phi)$ be, as before, the form that the observable takes, once it is re-expressed in terms of the variables $\tilde u,\phi$. 
We denote its formal Taylor expansion in $T$ by 
\begin{equation} \label{exp2} \tilde A(\tilde u,\phi)\stackrel{(\infty)}{=}\sum_{s\ge 0}T^s\sum_{p,\tilde p}a_{s}^{p,\tilde p}\phi^p\tilde u^{\tilde p}, \end{equation}
where $\stackrel{(\infty)}{=}$ denotes identity between formal power series, or between a function and its formal Taylor series,
the sums over $p,\tilde p$ run over $\mathfrak P_L, \tilde{\mathfrak P}_L$, respectively (see definition after \eqref{supp}), and $a_s^{p,\tilde p}$ are suitable real coefficients, independent of $T$. 
The formal low-$T$ expansion for 
\begin{equation}\label{eq:12}\langle a\rangle_{L;h}=\frac{\int\limits_{(\sqrt{\beta}\bbB^{N-2})^{\Lambda_L}}\!\!\! \!d{\tilde u}\, e^{\mathcal H_{L;h}(\tilde u)}\!\!\!\!
  \int\limits_{[-\sqrt{\beta}\pi,\sqrt{\beta}\pi]^{\Lambda_L}}\!\!\!\!d\phi \, \tilde A(\tilde u,\phi) e^{-\frac12\|\nabla\phi\|^2
  }e^{W_{L;h}(\tilde u,\phi)}}{\int\limits_{(\sqrt{\beta}\bbB^{N-2})^{\Lambda_L}}\!\!\! \!d{\tilde u}\, e^{\mathcal H_{L;h}(\tilde u)}\!\!\!\!
  \int\limits_{[-\sqrt{\beta}\pi,\sqrt{\beta}\pi]^{\Lambda_L}}\!\!\!\!d\phi \, e^{-\frac12\|\nabla\phi\|^2
  }e^{W_{L;h}(\tilde u,\phi)}}\end{equation}
is obtained by neglecting the constraints $|\tilde u_x|\le \sqrt{\beta}$, $|\phi_x|\le \pi\sqrt{\beta}$ in  the integration  domain (as we will show in the next section, this approximation 
produces an exponentially small error in  $\beta$, as $\beta\to\infty$), and by replacing \eqref{eq:8} and \eqref{exp2} in \eqref{eq:12}. In view of this, we write
\begin{equation}\label{eq:13}\langle a\rangle_{L;h}\stackrel{(\infty)}{=} \sum_{s\ge 0}T^s\sum_{p,\tilde p}a_{s}^{p,\tilde p}\frac{\int\!d{\tilde u}\, e^{\mathcal H_{L;h}(\tilde u)+
\mathcal V_{L;h}(\tilde u)}\, \tilde u^{\tilde p} \langle \phi^p\rangle_{L;h,\tilde u}}
{\int\!d{\tilde u}\, e^{\mathcal H_{L;h}(\tilde u)+
\mathcal V_{L;h}(\tilde u)}}
\end{equation}
where the integral over $\tilde u$ is performed over $(\mathbb R^{N-2})^{\Lambda_L}$, we defined
$\mathcal V_{L;h}(\tilde u)=\log$  $\int d\phi\,  e^{-\frac12\|\nabla\phi\|^2
+W_{L;h}(\tilde u,\phi)}$ (where the integral over $\phi$ is performed over $(\mathbb R^{\Lambda_L})^{\Lambda_L}$), and 
\begin{equation}
 \langle \phi^p\rangle_{L;h,\tilde u}=\int d\phi\, e^{-\frac12\|\nabla\phi\|^2
 +W_{L;h}(\tilde u,\phi)-\mathcal V_{L;h}(\tilde u)}\phi^p.
\end{equation}
By formally expanding in $W_{L;h}$, this expectation can be rewritten as
\begin{equation}  \langle \phi^p\rangle_{L;h,\tilde u}\stackrel{(\infty)}{=} \sum_{k\ge 0}\tfrac1{k!} \Phi_{L}(\phi^p;\underbrace{W_{L;h};\cdots;W_{L;h}}_{k\ \text{times}}), \end{equation}
where $\Phi_{L}(A_1;\cdots;A_k)$ denotes \textit{truncated} expectation with respect to the Gaussian law $\Phi_L$ with covariance $G_L$, 
that is, 
\begin{equation}\label{Urs1} \Phi_{L}(A_1;\cdots;A_k):=\partial_{\lambda_1}\cdots\partial_{\lambda_k}\log \Big(\Phi(
e^{\sum_{i=1}^k\lambda_i A_i})\Big)\Big|_{\lambda_0=\cdots=\lambda_k=0}\end{equation}
or, equivalently, 
\begin{equation}\label{Urs2}
\Phi_{L}(A_1;\cdots;A_k) = \sum_{\pi\in \calP_k} (|\pi|-1)! (-1)^{|\pi|-1} \prod_{Y\in \pi} \Phi_{L}(A_Y),
\end{equation} 
where \(\calP_k\equiv \calP([k])\) is the set of partitions of \([k]=\{1,\cdots, k\}\), and \(A_Y=\prod_{i\in Y} A_i\).
Inserting the formal low-$T$ expansion of $W_{L;h}$, \eqref{eq:8}, and taking $L\to\infty$, $h\searrow 0$, we find: 
\begin{equation}\label{eq:17}\begin{split} 
 \langle a\rangle \stackrel{(\infty)}{=} &\sum_{s_0\ge 0}T^{s_0}\sum_{p,\tilde p}a_{s_0}^{p,\tilde p}\,\Big[
   \Phi(\phi^p)\langle \tilde u^{\tilde p}\rangle+\sum_{s\ge 1}T^s \sum_{k\ge 1}\tfrac1{k!}\sum_{\substack{p_1,r_1,\ldots,p_k,r_k\ge 0:\\ p_l+r_l=s_l\ge 1\\
 s_1+\cdots+s_k=s}}\ \sum_{x_1,e_1,\ldots,x_k,e_k}\\
 &\hskip-.7truecm\cdot c_{r_1}\cdots c_{r_k} \Phi\big(\phi^p;(\nabla^{e_1}_{\!x_1}\phi)^{2r_1+2};\cdots;(\nabla^{e_k}_{\!x_k}\phi)^{2r_k+2}\big)
 \langle \tilde u^{\tilde p}\prod_{l=1}^k\mathcal G^{p_l}_{x_l,e_l}\rangle\Big],
 \end{split}
\end{equation}
where $\Phi$ is the law of the massless Gaussian Free Field on $\mathbb Z^d$,  the sum over $p,\tilde p$ can be freely restricted to even tuples (terms with $p$ or $\tilde p$ odd are zero by parity), and the sums over $x_1,\ldots,x_k$ are now over $\mathbb Z^d$.

\subsection{The finite order truncations of the low temperature expansion}

If we denote by $ \langle a\rangle^{(n)}$ the truncation at order $n$ of the formal Taylor series in $T$, \eqref{eq:17} implies that 
\begin{equation}\label{eq:18}\begin{split} 
& \langle a\rangle^{(n)}=\sum_{s_0=0}^n T^{s_0}\sum_{p,\tilde p}a_{s_0}^{p,\tilde p}\,\Big[
   \Phi(\phi^p)\langle \tilde u^{\tilde p}\rangle^{(n-s_0)}+\mathds 1_{n>s_0}\mathds 1_{\|p\|_1>0}\sum_{s=1}^{n-s_0}T^s \sum_{k\ge 1}\tfrac1{k!}\sum_{\substack{p_1,r_1,\ldots,p_k,r_k\ge 0:\\ p_l+r_l=s_l\ge 1\\
 s_1+\cdots+s_k=s}}\cdot\\ 
 &\qquad \cdot c_{r_1}\cdots c_{r_k}  \sum_{x_1,e_1,\ldots,x_k,e_k}\Phi\big(\phi^p;(\nabla^{e_1}_{\!x_1}\phi)^{2r_1+2};\cdots;(\nabla^{e_k}_{\!x_k}\phi)^{2r_k+2}\big)
 \langle \tilde u^{\tilde p}\prod_{l=1}^k\mathcal G^{p_l}_{x_l,e_l}\rangle^{(n-s-s_0)}\Big],
 \end{split}
\end{equation}
where, again, the sum over $p,\tilde p$ can be restricted to even  tuples. 
Now, the terms $\langle \tilde u^{\tilde p}\rangle^{(n-s_0)}$  in the right side
(and similarly $\langle \tilde u^{\tilde p}\prod_{l=1}^k\mathcal G^{p_l}_{x_l,e_l}\rangle^{(n-s-s_0)}$, which are linear combinations of terms of the same form, i.e., of $\langle \tilde u^{\tilde p'}\rangle^{(n-s-s_0)}$ for suitable $\tilde p'$), can be 
computed recursively, by using \eqref{eq:remaining_symmetry_u_theta} and \eqref{eq:18} itself for a finite number of times. In fact, let $m+\|\tilde p\|_1/2\ge 0$ be the \textit{order}
of $\langle \tilde u^{\tilde p}\rangle^{(m)}$, where $\tilde p$ is an even tuple with non-negative elements. 
If the order is zero, then $\langle 1\rangle^{(0)}=1$. More generally, if $\tilde p=0$, $\langle 1\rangle^{(m)}=1$, for all $m\ge 0$. Let us then consider a term 
$\langle \tilde u^{\tilde p}\rangle^{(m)}$ of positive order $k=m+\|\tilde p\|_1/2$, with $\|\tilde p\|_1$ positive. 
Let us assume without loss of generality that $\|\tilde p^1\|$ is even and positive, so that $\tilde p^{\ge 2}=(0,\tilde p^2,\ldots,\tilde p^{N-2})$ has norm $\|\tilde p^{\ge 2}\|_1\le \|\tilde p\|_1-2$. 
Using \eqref{eq:remaining_symmetry_u_theta}, and expanding $\prod_x\Big(\sqrt{\beta(1-T\|\tilde u_x\|^2)}\sin\big(\sqrt{T}\phi_x\big)\Big)^{\tilde p^1_x}$ in Taylor series in $T$ in the form 
\begin{equation}\label{expT}\begin{split} &\phi^{\tilde p^1}\sum_{s\ge 0}(-1)^sT^s\sum_{\substack{(s_x)_{x\in\mathbb Z^d}:\\ s_x\ge 0,\ \sum_xs_x=s}}\prod_{x\in\mathbb Z^d}\sum_{\substack{i_x,j_x\ge 0: \\ i_x+j_x=s_x}}\phi_x^{2i_x}\|\tilde u_x\|^{2j_x}\frac1{(2i_x+1)!}{1/2\choose j_x}\\
 \equiv\ & \phi^{\tilde p^1}\sum_{s\ge 0}T^s \sum_{n\in\mathfrak P}\sum_{\tilde n\in\tilde{\mathfrak P}} c^{n,\tilde n}_s \phi^n \tilde u^{\tilde n}\end{split}\end{equation}
 (note that $c^{n,\tilde n}_s$ is non zero only if $n,\tilde n$ are both even and $\|n\|_1+\|\tilde n\|_1=2s$), we find
\begin{equation} \langle \tilde u^{\tilde p}\rangle^{(m)}= \sum_{s=0}^m T^s \sum_{\substack{n\in\mathfrak P\\ \tilde n\in\tilde{\mathfrak P}}}c_s^{n,\tilde n}\langle \phi^{\tilde p^1+n}\tilde u^{\tilde p^{\ge 2}+\tilde n}\rangle^{(m-s)}.\end{equation}
By applying formula \eqref{eq:18} to the terms $\langle \phi^{\tilde p^1+n}\tilde u^{\tilde p^{\ge 2}+\tilde n}\rangle^{(m-s)}$ in the right side, we can re-express each of them as linear combinations of terms $\langle \tilde u^{\tilde p'}\rangle^{(m')}$, whose order, $k'=m'+\|\tilde p'\|_1/2$, is strictly smaller than the original order $k$. Therefore, by proceeding recursively, 
we can compute $\langle \tilde u^{\tilde p}\rangle^{(m)}$ explicitly. By plugging the result back into \eqref{eq:18} we obtain the (truncation at order $n$ of the) low temperature expansion 
of $\langle a\rangle$, for any observable $a$. 

\medskip

For later reference, we formalize the definition of $\lrangle{\ }^{(n)}$ as follows. Let $\mathcal A$ be the space of formal power series in $T$ of the form 
$\sum_{s\ge 0}T^s\sum_{p\in\mathfrak P}\sum_{\tilde p\in\tilde{\mathfrak P}}a^{p,\tilde p}_s\phi^p\tilde u^{\tilde p}$ with real coefficients $a^{p,\tilde p}_s$ independent of $T$. Let also $\mathcal A_0$ be the subspace of
$\mathcal A$ consisting of formal power series independent of $\tilde u,\phi$, of the form $\sum_{s}T^s a_s$, with $a_s\in\mathbb R$.

\begin{definition}\label{def:3.1} For any $n\ge 0$, \(\lrangle{ \ }^{(n)}\) is the operator acting on $\mathcal A$ that satisfies the following properties. 
	\begin{itemize}
	\item  \(\lrangle{1}^{(n)} = 1\).
	\item \(\lrangle{\ }^{(n)}\) is translation invariant, that is, for any $p\in\mathfrak P$, $\tilde p\in\tilde{\mathfrak P}$, $x_0\in\mathbb Z^d$, 
	$\lrangle{\phi^p\tilde u^{\tilde p}}^{(n)}=\lrangle{\phi^{\tau_{x_0}p}\tilde u^{\tau_{x_0}\tilde p}}^{(n)}$, where $(\tau_{x_0}p)_x=p_{x-x_0}$, and similarly for $\tau_{x_0}\tilde p$. 
	\item For any $Q,Q'\in\mathcal A$, \(\lrangle{ Q+Q' }^{(n)} = \lrangle{ Q }^{(n)} + \lrangle{ Q' }^{(n)}\).
	\item For any $\lambda\in\mathcal A_0$ of the form $\lambda=\sum_{s\ge 0}T^s \lambda_s$, and any $Q\in \mathcal A$, 	
	\(\lrangle{ \lambda Q}^{(n)} = \sum_{s=0}^nT^s \lambda_s \lrangle{  Q}^{(n-s)} \). 
	\item For any $p\in\mathfrak P$ and any $\tilde p\in\tilde{\mathfrak P}$, letting $\mathcal F=\mathcal F(\tilde u^1,\ldots,\tilde u^{N-2})=\tilde u^{\tilde p}$, 
	\begin{equation}\label{eq:symmdef}\lrangle{ \phi^p \calF(\tilde{u}^1,\cdots,\tilde{u}^{N-2}) }^{(n)} = \lrangle{ \phi^p \calF(\tilde{u}^{\pi(1)},\cdots,\tilde{u}^{\pi(N-2)}) }^{(n)}
	\end{equation}
	for any permutation \(\pi\) of \(\{1,\cdots, N-2\}\).
	\item For any $\tilde p\in\tilde{\mathfrak P}$, if $\tilde p^1=(\tilde p^1_x)_{x\in\mathbb Z^d}\in \mathfrak P$ is not zero, letting $\tilde p^{\ge 2}=(0,\tilde p^2,\ldots,\tilde p^{N-2})$, 
	\begin{equation}\label{eq:symm2def}
		\lrangle{ \tilde u^{\tilde p}}^{(n)} = \lrangle{\phi^{\tilde p^1} \tilde{F}  \tilde u^{\tilde p^{\ge 2}} }^{(n)}
	\end{equation}where \(\tilde{F}=\sum_{s\ge 0}T^s\sum_{n,\tilde n}c^{n,\tilde n}_s\phi^n\tilde u^{\tilde n}\in\mathcal A$ (the coefficients $c^{n,\tilde n}_s$ were defined 
	in \eqref{expT}).
	\item For any even $p\in\mathfrak P$, $\tilde p\in\tilde{\mathfrak P}$, the following ``extraction'' formula holds:
	\begin{equation} \label{eq:phi_contrib_extraction}\begin{split}
		\lrangle{\phi^p\tilde u^{\tilde p}}^{(n)} &= \Phi(\phi^p)\lrangle{\tilde u^{\tilde p}}^{(n)}  + \mathds{1}_{n\geq 1} \sum_{k=1}^{n}\frac{1}{k!} \sum_{s=k}^n T^{s}  \sum_{\substack{s_1,\cdots,s_k \geq 1\\ \sum s_l = s}}\sum_{\substack{x_1,\cdots, x_k\\ e_1,\cdots,e_k}}\sum_{\substack{r_1,p'_1,\cdots, r_k,p'_k\ge 0\\ p'_l+r_l = s_l} }\times \\ & \times  \big(\prod_{l=1}^k c_{r_l}\big) \Phi\big( \phi^p ; (\nabla_{x_1}^{e_1}\phi)^{2r_1+2};\cdots; (\nabla_{x_k}^{e_k}\phi)^{2r_k+2}\big)\lrangle{\tilde u^{\tilde p} \prod_{l=1}^{k}\calG_{x_l,e_l}^{p_l'}}^{(n-s)},
	\end{split}\end{equation}
where $c_r$ and $\mathcal G^{p}_{x,e}$ were defined right after \eqref{eq:8}. 
\end{itemize}
\end{definition}

Note that the presence of the {\it connected} correlation $\Phi\big( \phi^p ; (\nabla_{x_1}^{e_1}\phi)^{2r_1+2};\cdots; (\nabla_{x_k}^{e_k}\phi)^{2r_k+2}\big)$ 
ensures absolute summability of the second line of 
\eqref{eq:phi_contrib_extraction}, see Lemma \ref{app:lem:sum_of_connected_corr}, thus making \(\lrangle{\ }^{(n)}\) a well defined object. More precisely, Lemma \ref{app:lem:sum_of_connected_corr} implies the following quantitative bounds, which will be used in the inductive proof discussed in Section \ref{sec:Inductive_exp}. 

\begin{lemma}
	\label{lem:Taylor_n_a_priori_decay_no_grad}
	Let \(K>0, n\geq 0\) be integers. Let \(\epsilon>0\). Let \(p,p':\Zd\to \Z_+\) be odd, \(\tilde{p}: \Zd \times[N-2]\to \Z_+\) be even, and such that \(\normI{p}+\normI{p'}+\normI{\tilde{p}} \leq 2K\), then,
	\begin{equation*}
		|\lrangle{\phi^p \phi^{p'} \tilde{u}^{\tilde{p}}}^{(n)}| \leq C \sum_{x\in \supp_p}\sum_{y\in\supp_{p'}} \frac{1}{(1+|x-y|)^{d-2-\epsilon}},
	\end{equation*}where \(C =C(K,n,\epsilon)\).
\end{lemma}
\begin{lemma}
	\label{lem:Taylor_n_a_priori_decay_grad}
	Let \(K>0, n\geq 0\) be integers. Let \(\epsilon>0\). Let \(p:\Zd\to \Z_+\) be odd, \(q:\bbE_d\to \Z_+\) be odd, \(\tilde{p}: \Zd \times[N-2]\to \Z_+\) be even, and such that \(\normI{p}+ \normI{q}+\normI{\tilde{p}} \leq 2K\), then,
	\begin{equation*}
		|\lrangle{\phi^p (\nabla\phi)^{q} \tilde{u}^{\tilde{p}}}^{(n)}| \leq C \sum_{x\in \supp_p}\sum_{y\in\supp_{q}} \frac{1}{(1+|x-y|)^{d-1-\epsilon}},
	\end{equation*}where \(C =C(K,n,\epsilon)\).
\end{lemma}

These two lemmas follows straightforwardly from Lemma~\ref{app:lem:sum_of_connected_corr} and the definition of \(\lrangle{\ }^{(n)}\). 
In fact, notice first of all that eq.\eqref{eq:71} of Lemma~\ref{app:lem:sum_of_connected_corr} implies that, for any even $\tilde p\in\tilde {\mathfrak P}$, 
\begin{equation}\label{boundtildeu}|\lrangle{\tilde u^{\tilde p}}^{(n)}|\le C,\end{equation} with $C=C(\|\tilde p\|_1,n)$: this can be easily proved by induction in $n+\|\tilde p\|_1/2$, proceeding as described after 
\eqref{eq:18}. More in detail: consider a non zero even $\tilde p$ and $n\ge 0$; assume inductively that $|\lrangle{\tilde u^{\tilde p'}}^{(n')}|\le C$ is known for all
$\tilde p',n'$ such that $n'+\|\tilde p'\|_1/2<n+\|\tilde p\|_1/2$; use \eqref{eq:symmdef}  to reduce to the case of $\tilde p^1$ non zero; apply \eqref{eq:symm2def}, then use \eqref{eq:phi_contrib_extraction}; write $\tilde p^1=p+p'$ with $p,p'$ odd, then 
use \eqref{eq:71} and the inductive hypothesis to conclude the proof of the induction step. At this point, the two lemmas follow immediately from 
 \eqref{eq:phi_contrib_extraction}, the use of \eqref{eq:71} and \eqref{eq:72}, respectively, and of \eqref{boundtildeu} (details left to the reader).
 
\section{Applications of the infrared bound}
\label{sec:infraref_bnd}

In this section we describe an important application of the infrared bound summarized in Proposition \ref{thm:Infrared_infinite_volume_magn_bound} above. In particular, 
generalizing a method in \cite{Bricmont+Fontaine+Lebowitz+Lieb+Spencer-1980}, we derive a priori bounds on the moments of the spin correlation functions, in the $(u,\theta)$, or equivalently $(\tilde u,\phi)$, coordinates. These will be used in the following sections in order to control the remainder of the low temperature expansion, as produced by systematic
integration by parts.

\subsection{Large deviation estimates on the spin correlations.}

\begin{lemma}
	\label{lem:moment_bnd_spins}
For any $\beta\ge 0$, $h\in D^*$ and $a>0$, if \(k\in\{1,\cdots,N-1\}\),
	\begin{equation}
	\label{eq:Moment_bound_spin_version_transverse}
	\lrangle{e^{a\sqrt{\beta}|S_0^k|}}_h\leq 2e^{a^2G_{00}/2}.
	\end{equation}
	Moreover,
	\begin{equation}
	\label{eq:Moment_bound_spin_version_longi}
	\lrangle{e^{a\sqrt{\beta}(1-S_0^N)}}_h \leq e^{(a\sqrt{T}N+a^2/2)G_{00}}.
	\end{equation}
\end{lemma}
\begin{proof}
	By the infrared bound~\eqref{eq:infrared_infin_vol_centered_spins} applied to the function \(f(x) = -a\sqrt{\beta}\delta_{0,x} \rme_k\), one obtains
	\begin{equation}\label{eq:senza}
	\lrangle{e^{-a\sqrt{\beta} (S_0^k-\lrangle{S_0^k}_h)}}_h \leq e^{a^2G_{00}/2}.
	\end{equation}Now, for \(k=1,\cdots,N-1\), by symmetry \(S_0^k\) has the same law as \(-S_0^{k}\). In particular, \(\lrangle{S_0^k}_h =0\) and
	\begin{equation*}
	\lrangle{e^{a\sqrt{\beta}|S_0^k|}}_h = \lrangle{e^{a\sqrt{\beta}S_0^k} \mathds{1}_{S_0^k\geq 0} }_h+ \lrangle{e^{-a\sqrt{\beta}S_0^k} \mathds{1}_{S_0^k< 0} }_h
	\leq 2\lrangle{e^{-a\sqrt{\beta}S_0^k} }_h.
	\end{equation*}Thus~\eqref{eq:Moment_bound_spin_version_transverse} holds for \(k=1,\cdots,N-1\). Notice that this bound holds for any value of \(T\).
	For \(k=N\), the lower bound~\eqref{eq:infrared_magnet} on \(\lrangle{S_0^N}_h\), and \(\lrangle{S_0^N}_h\geq 0\) implies that
	\begin{equation*}
		\lrangle{S_0^N}_h \geq  \sqrt{\max(1 - TN G_{00}, 0)} \geq \max(1 - TN G_{00}, 0).
	\end{equation*}
	Plugging this bound in \eqref{eq:senza} with $k=N$, we find
	\begin{equation*}
		e^{a^2G_{00}/2}\geq e^{a\sqrt{\beta}\lrangle{S_0^N}_h}\lrangle{e^{-a\sqrt{\beta}S_0^N}}_h \geq e^{-a\sqrt{T} N G_{00}}\lrangle{e^{a\sqrt{\beta}(1-S_0^N)}}_h
	\end{equation*}so that
	\begin{equation*}
	\lrangle{e^{a\sqrt{\beta}(1-S_0^N)}}_h \leq e^{a^2G_{00}/2} e^{a\sqrt{T} N G_{00} } \leq e^{(a\sqrt{T}N+a^2/2)G_{00}},
	\end{equation*} proving~\eqref{eq:Moment_bound_spin_version_longi}.
\end{proof}

A direct consequence is
\begin{corollary}
Let $\beta\ge 0$, $h\in D^*$ and $b>0$. For any $1\le k_0\le N-1$, 
	\begin{equation}
	\label{eq:sign_Sk_non_neg}
		\mu_h\big(\sum_{k=1}^{k_0}|S_0^k|^2 \geq b^2\big) \leq 2 e^{- b^2\beta/(2k_0^2G_{00})}.
	\end{equation}
If, additionally, $T\le b/(2G_{00}N)$, then 
	\begin{equation}
	\label{eq:sign_SN_non_neg}
		\mu_h(S_0^N \leq 1-b) \leq e^{- b^2\beta/(8 G_{00})}.
	\end{equation}
\end{corollary}
\begin{proof}
	Using the exponential version of Chebyshev's inequality, H\"older's inequality, and 
	\eqref{eq:Moment_bound_spin_version_transverse}, we find that, for any $a>0$, 
\begin{equation*}\begin{split}
		\mu_h(\sum_{k=1}^{k_0}|S_0^k|^2 \geq b^2) & \leq \mu_h(\sum_{k=1}^{k_0}|S_0^k| \geq b) \\
		&\leq e^{-ab\sqrt{\beta}} \lrangle{e^{a\sqrt{\beta}\sum_{k=1}^{k_0}|S_0^k|}}_h\\
		&\le e^{-ab\sqrt{\beta}}\prod_{k=1}^{k_0} \lrangle{e^{ak_0\sqrt{\beta}|S_0^k|}}^{1/k_0}_h\\
                 &\leq 2 e^{-ab\sqrt{\beta}}e^{a^2k_0^2G_{00}/2}.\end{split}
	\end{equation*}
Choosing $a=b\sqrt{\beta}/(k_0^2G_{00})$ implies \eqref{eq:sign_Sk_non_neg}. Next, using 
\eqref{eq:Moment_bound_spin_version_longi} and, again, the exponential version of Chebyshev's inequality, for any $a>0$,
	\begin{equation*}
		\mu_h (1-S_0^N \geq b) \leq e^{-ab\sqrt{\beta}} \lrangle{e^{a\sqrt{\beta}(1-S_0^N)}}_h\leq e^{-ab\sqrt{\beta}}e^{(a\sqrt{T}N+a^2/2) G_{00}}.
	\end{equation*}
Choosing $a=b\sqrt{\beta}/G_{00}-\sqrt{T} N$ and $T\le b/(2G_{00}N)$ implies 	\eqref{eq:sign_SN_non_neg}. 
\end{proof}

\subsection{Bounds on the moments of the \((u,\theta)\) coordinates}

The next Theorem is the main objective of this section. Its proof is a generalization of the one of the corresponding result in~\cite{Bricmont+Fontaine+Lebowitz+Lieb+Spencer-1980}, 
namely \cite[Eq.(10) and Lemma 2]{Bricmont+Fontaine+Lebowitz+Lieb+Spencer-1980}. Here and below, with some abuse of notation, we denote by $\lrangle{\ }_h$ the average 
with respect to $\mu_h$, even when re-expressed in terms of the $(u,\theta)$ variables, rather than of $S$ (a similar convention will be used for $\lrangle{\ }_{L;h}$ and $\lrangle{\ }$). 

\begin{theorem}
	\label{thm:moments_u_theta}
	There exist constants \(C,C_0,c>0\) and \(T_0>0\) such that, if \(T\le T_0\), the following holds.
	\begin{itemize}
		\item For any \(k\in\{1,\cdots, N-2\}\) and $a>0$, \(\lrangle{e^{a|\tilde{u}_0^k|}}_h\leq 2e^{a^2G_{00}/2}\). Moreover, for any \(n\geq 0\), \begin{equation*}
		\lrangle{|\tilde{u}^k_0|^n}_h\leq C^n\sqrt{n!}.
		\end{equation*}
		\item \(\lrangle{e^{ |\phi_0|}}_h\leq C\). In particular, for any \(n\geq 0\), \begin{equation*}
		|\lrangle{|\phi_0|^n}_h|\leq C n!.
		\end{equation*}
		\item For any \(F\) function of \(\phi,\tilde{u}\), \begin{equation*}
		|\lrangle{F\delta(\theta_0\pm \pi)}_h| \leq \normsup{F}C_0e^{h-c\beta},
		\end{equation*}where the sup-norm is computed over \(\phi_x\in[-\sqrt{\beta}\pi,\sqrt{\beta}\pi ]\) and \(\norm{\tilde{u}_x}\leq \sqrt{\beta}\) for all \(x\in\Zd\).
	\end{itemize}
\end{theorem}
\begin{proof}
	The first inequality in the first item, \(\lrangle{e^{a|\tilde{u}_0^k|}}_h\leq 2e^{a^2G_{00}/2}\), is a direct consequence of~\eqref{eq:Moment_bound_spin_version_transverse} and the definition of \(\tilde{u}\). From such inequality, it follows that $\lrangle{|\tilde{u}_0^k|^n}_h\leq 2e^{a^2G_{00}/2}a^{-n}n!$. The choice $a=\sqrt{n/G_{00}}$ implies the stated bound on 
$\lrangle{|\tilde{u}_0^k|^n}_h$. 
	
	To get the second item, write
	\begin{align*}
		\lrangle{e^{\sqrt{\beta}|\theta_0|}}_h &= \lrangle{e^{\sqrt{\beta}|\theta_0|} \mathds{1}_{|\theta_0|\leq \pi/2}\mathds{1}_{\|u_0\|\leq 1/2}
		}_h + \lrangle{e^{\sqrt{\beta}|\theta_0|} \mathds{1}_{|\theta_0|\leq \pi/2}\mathds{1}_{\|u_0\|> 1/2}
		}_h +		
		\lrangle{e^{\sqrt{\beta}|\theta_0|}\mathds{1}_{|\theta_0|> \pi/2}}_h\\
		&\leq \lrangle{e^{\pi^{-1}\sqrt{3\beta}\sqrt{1-\|u_0\|^2}|\sin(\theta_0)|}}_h + e^{ \pi\sqrt{\beta}/2} \mu_h(\|u_0\|>1/2)+ e^{ \pi\sqrt{\beta}}\mu_h(|\theta_0|> \pi/2)\\
		&= \lrangle{e^{\pi^{-1}\sqrt{3\beta}|S_0^{N-1}|}}_h + e^{\pi\sqrt{\beta}/2}\mu_h\big({\sum_{k=1}^{N-2}|S^k_0|^2}>1/4\big)+e^{\pi\sqrt{\beta}} \mu_h(S_0^N \leq 0)\\ 
		&\leq 2e^{3G_{00}/(2\pi^2)} + 2 e^{\pi\sqrt{\beta}/2}e^{-\beta/(8G_{00}(N-2)^2)}+e^{\pi\sqrt{\beta}}e^{-\beta/(8G_{00})},
	\end{align*}
where in the last inequality we used \eqref{eq:Moment_bound_spin_version_transverse}, \eqref{eq:sign_Sk_non_neg} and~\eqref{eq:sign_SN_non_neg}. If we now choose $\beta$ large enough and recall the definition of $\phi$, we obtain, as desired, that 
$\lrangle{e^{|\phi_0|}}_h\le C$ for a suitable $C>0$. 

 Finally, we consider the term $\lrangle{F\delta(\theta_0\pm \pi)}_h$ involving the delta function. Using the DLR equation and letting \(S_0(u_0)=\Pi(u_0,\pi)=(u_0^1,\cdots, u_0^{N-2}, 0, -\sqrt{1-\norm{u_0}^2} )\),
	\begin{equation*}
		\lrangle{F\delta(\theta_0\pm \pi)}_h = \int d\mu_h(S_i,i\neq 0) \frac{\int_{\norm{u_0}\leq 1} du_0\, F(S,S_0(u_0))e^{\beta \sum_{i\sim 0} S_0(u_0)\cdot S_i -\sqrt{1-\norm{u_0}^2} h} }{\int d\nu_N(\tilde{S}_0) e^{\beta \sum_{i\sim 0} \tilde{S}_0\cdot S_i + h\tilde{S}_0^N}}.
	\end{equation*}Define
	\begin{equation*}
		J( S_i, i\sim 0) =\sup_{\norm{u_0}\leq 1}\Big(\int d\nu_N(\tilde{S}_0) e^{\beta \sum_{i\sim 0} (\tilde{S}_0-S_0(u_0))\cdot S_i + h\tilde{S}_0^N}\Big)^{-1}\equiv
		 \sup_{\norm{u_0}\leq 1} J_{u_0}( S_i, i\sim 0).
	\end{equation*}One has
	\begin{equation*}
		|\lrangle{F\delta(\theta_0\pm \pi)}_h|\leq \normsup{F}\frac{\pi^{N/2-1}}{\Gamma(N/2)} \int d\mu_h(S_i,i\neq 0) J( S_i, i\sim 0).
	\end{equation*} Now, for any event \(A\) measurable with respect to \(\{S_i: i\sim 0\}\), one has
	\begin{equation}
	\label{eq:moments_u_theta_eq1}
		\int d\mu_h(S_i,i\neq 0) J( S_i, i\sim 0) \leq \sup_{(S_i)_{i\sim 0}\in A} J( S_i, i\sim 0) + \mu_h(A^c)\sup_{(S_i)_{i\sim 0}\in A^c} J( S_i, i\sim 0).
	\end{equation}We choose \(A = \{S_i^N\geq 7/8, i\sim 0\}\). Study the first term. For any \((S_i)_{i\sim 0}\in A\), any $u_0\in\mathbb B^{N-2}$ and $h\ge 0$, one has the bound
	\begin{equation*}\begin{split}
		J_{u_0}( S_i, i\sim 0) &\leq  \Big(\int d\nu_N(\tilde{S}_0) \mathds{1}_{\tilde{S}_0^N\geq 7/8}e^{\beta \sum_{i\sim 0} (\tilde{S}_0-S_0(u_0))\cdot S_i}\Big)^{-1} \\ 
		&\leq \Big(\int d\nu_N(\tilde{S}_0) \mathds{1}_{\tilde{S}_0^N\geq 7/8}e^{\beta 2d/32}\Big)^{-1} = Ce^{-\beta 2d/32}.\end{split}
	\end{equation*} for some constant $C>0$. Indeed, on the one hand, for $S_i^N\ge 7/8$ and $\|u_0\|\le 1$, by Cauchy-Schwartz,
	\begin{equation*}\begin{split}
		S_0(u_0)\cdot S_i &= \sum_{k=1}^{N-2} u_0^kS_i^k - \sqrt{1-\norm{u_0}^2} S_i^N \leq \norm{u_0}\Big(\sum_{k=1}^{N-1} (S_i^k)^2\Big)^{1/2} 
		\\
		& \leq \norm{u_0}\Big(1-\frac{49}{64}\Big)^{1/2} \leq \norm{u_0}\cdot\frac{1}{2}\leq \frac{1}{2}.\end{split}
	\end{equation*}
	On the other hand, if \(\tilde{S}^N_0\geq 7/8\) and $S_i^N\ge 7/8$,
	\begin{equation}
		\tilde{S}_0\cdot S_i \geq \cos(2\arccos(7/8)) = \frac{17}{32}.
	\end{equation}
	We now turn to the second term in~\eqref{eq:moments_u_theta_eq1}. First, by a union bound, 
	\begin{equation*}
		\mu_h(A^c)\leq 2d \mu_h(1-S_0^N> 1/8) \leq 2d e^{- \beta/(8^3 G_{00})},
	\end{equation*}where we assumed $T$ sufficiently small and used~\eqref{eq:sign_SN_non_neg}. Moreover, for any $u_0\in\mathbb B^{N-2}$, 
	\begin{equation*}\begin{split}
		J_{u_0}( S_i, i\sim 0) & \leq e^h\Big(\int d\nu_N(\tilde{S}_0) \mathds{1}_{\norm{\tilde{S}_0-S_0(u_0)}\leq T}e^{\beta \sum_{i\sim 0} (\tilde{S}_0-S_0(u_0))\cdot S_i}\Big)^{-1}\\
		& \leq e^{h+2d}\Big(\int d\nu_N(\tilde{S}_0) \mathds{1}_{\norm{\tilde{S}_0-S_0(u_0)}\leq T}\Big)^{-1}\leq C'e^h\beta^{N-1},\end{split}
	\end{equation*}for some \(C'>0\) independent of \(T\). Combining the two previous equations, we obtain the desired estimate on the second term in~\eqref{eq:moments_u_theta_eq1}, which concludes the proof.
\end{proof}

\subsection{Taylor expansion of the spin correlations} Given \(\alpha:\Zd\times [N]\to \Z_+\) with \(\|\alpha\|_1=\sum_{k=1}^N\sum_{x\in \Zd}\alpha^k_x\) finite, 
we let  $a_\alpha(S)$ be the spin observable 
$$a_\alpha(S)=\prod_{x}\prod_k (S_x^k)^{\alpha_{x}^k}.$$
In order to prove our main result, Theorem \ref{thm:main}, we are interested in computing $\lrangle{a_\alpha(S)}$ for $\alpha$s such that $\|\alpha^k\|_1=\sum_{x}\alpha^k_x$ are even, for all $k\in\{1,\ldots,N-1\}$ 
(if any of these $\|\alpha^k\|_1$ is odd, $\lrangle{a_\alpha(S)}=0$, due to the residual $O(N-1)$ symmetry in the directions orthogonal to the magnetic field).  
The main result of this subsection is the following:
\begin{lemma}\label{lem:tay}
Let $\alpha$ and $a_\alpha(S)$ be defined as above. Then there exist coefficients $a_{\alpha,s}^{p,\tilde p}$, with $s\ge 0$ and $p\in \mathfrak P$, $\tilde p\in\tilde{\mathfrak P}$ both even, such that, for any $n\ge 0$, 
\begin{equation}\label{eq6.2.1}\lrangle{a_\alpha(S)}=\sum_{s=0}^nT^s \sum_{p,\tilde p} a_{\alpha,s}^{p,\tilde p}\lrangle{\phi^p\tilde u^{\tilde p}}+R_{n+1}(\alpha),\end{equation}
where the sum over $p,\tilde p$ in the right hand side runs over even tuples in $\mathfrak P,\tilde{\mathfrak P}$, $a_{\alpha,s}^{p,\tilde p} =0$ if $\|p\|_1+\|\tilde p\|_1\neq 2s$, and the remainder satisfies
\begin{equation}\label{eq:remai}|R_{n+1}(\alpha)|\le C(n,\alpha)T^{n+1}\end{equation}
for some $C(n,\alpha)>0$. 
\end{lemma}

\begin{remark}\label{rem:taylorr}The coefficients $a_{\alpha,s}^{p,\tilde p}$ are explicitly computable in terms of the Taylor expansions of $\sqrt{1-x}$, $\sin x$ and $\cos x$, see \eqref{exptutti}-\eqref{expAalpha} below.
\end{remark}

\begin{proof}
We denote by $A_\alpha(\tilde u,\phi)$ the rewriting of $a_\alpha(S)$ in the $(\tilde u,\phi)$ variables:
\begin{equation}\label{eq:reprAalpha}A_\alpha(\tilde u, \phi)=\prod_{x}\prod_k \Big(\Pi_k(\sqrt{T}\tilde u_x,\sqrt{T}\phi_x)\Big)^{\alpha_{x}^k},\end{equation}
with $\Pi_k$ the components of the function $\Pi$ in \eqref{transf}. Letting $E_\alpha$ be the event 
$$E_\alpha=\{\tilde u\in (\sqrt{\beta}\mathbb B^{N-2})^{\mathbb Z^d} : \|\tilde u_x\|\le\sqrt{\beta}/2, \forall x\in\supp\,\alpha^{N-1}\cup\supp\,\alpha^N\},$$
and using \eqref{eq:sign_Sk_non_neg}, which implies that $\mu(E_\alpha^c)=O(e^{-c\beta})$ for some $c>0$, we find
\begin{equation} \lrangle{a_\alpha(S)}=\lrangle{\mathds 1_{E_\alpha}\, A_\alpha(\tilde u,\phi)}+O(e^{-c\beta}).\end{equation}
By expanding each factor $\sqrt{1-T\|\tilde u_x\|^2}$, $\sin(\sqrt{T}\phi_x)$, $\cos(\sqrt{T}\phi_x)$ appearing in the right hand side 
of \eqref{eq:reprAalpha} in power series in $T$, as 
\begin{equation}\label{exptutti}\begin{split}
& \sqrt{1-T\|\tilde u_x\|^2}\stackrel{(\infty)}{=}\sum_{n\ge 0}T^n(-1)^n{1/2 \choose n}\|\tilde u_x\|^{2n},\\
& \sin(\sqrt{T}\phi_x)=\sqrt{T}\phi_x\sum_{n\ge 0}
T^n\frac{(-1)^n}{(2n+1)!}\phi_x^{2n},\\
& \cos(\sqrt{T}\phi_x)=\sum_{n\ge 0}
T^n\frac{(-1)^n}{(2n)!}\phi_x^{2n}, \end{split}\end{equation}
we obtain the formal power series expansion for $A_\alpha(\tilde u,\phi)$, in the form
\begin{equation}\label{expAalpha}A_\alpha(\tilde u, \phi)\stackrel{(\infty)}{=}\sum_{s\ge 0}T^s \sum_{p,\tilde p} a_{\alpha,s}^{p,\tilde p}\phi^p\tilde u^{\tilde p},\end{equation}
for suitable real coefficients $a^{p,\tilde p}_{\alpha,s}$, which are non zero only if $\|p\|_1+\|\tilde p\|_1=2s$ (the sum over $p,\tilde p$ runs over even tuples in $\mathfrak P,\tilde{\mathfrak P}$). By the Taylor remainder's theorem, for any $n\ge 0$, there exists $C_{n,\alpha}>0$ such that, for any $\tilde u$ in the support of $E_\alpha$, 
$$\Big|A_\alpha(\tilde u, \phi)-\sum_{s=0}^nT^s \sum_{p,\tilde p} a_{\alpha,s}^{p,\tilde p}\phi^p\tilde u^{\tilde p}\Big|\le C_{n,\alpha} T^{n+1}\sum_{x\in\supp\alpha}(\|\tilde u_x\|^{2(n+1)}+\phi_x^{2(n+1)}).$$
Therefore, using again \eqref{eq:sign_Sk_non_neg} and the moments bounds in the first two items of Theorem \ref{thm:moments_u_theta}, we obtain \eqref{eq6.2.1}-\eqref{eq:remai}, as desired. \end{proof}

\section{Integration by part}
\label{sec:int_by_part}

From now on, we shall assume \(T\le T_0\), with $T_0$ the same as in the statement of Theorem \ref{thm:moments_u_theta}, without recalling it each time. 
The main results of this section are the following two integration-by-part lemmas, which are the key ingredients of the inductive computation of the coefficients of the low temperature expansion for the correlation functions, described in Section \ref{sec:Inductive_exp}. 

\begin{lemma}[Integration by part: non-gradient]
	\label{lem:int_by_part_no_grad}
	Let \(p,p'\geq 0\) be integers. Let \(F\) be a monomial in \(\phi\) of degree \(2p+1\), and \(\calF\) be a monomial in \((\tilde{u}_x^k)_{x,k}\) with even degree at most \(2p'\), both with coefficient \(1\). Then, for any \(x_0\in\Zd\), integers \(\gamma>n\geq 0\), and \(\epsilon>0\),
	\begin{align*}
		\lrangle{\phi_{x_0}F\calF} &= \sum_{y\in\mathrm{supp}_F}  G(x_0,y)\lrangle{\partial_{y}F\calF} +\\
		&\ + \mathds 1_{n\ge 1}\sum_{k=1}^{n} T^{k} \sum_{x_1,e}\nabla_{x_1}^{\rme} G^{m_{\gamma}}(x_0,\cdot) \sum_{r+p' = k} c'_r \lrangle{F\calF (\nabla_{x_1}^{e}\phi)^{2r+1}\mathcal G^{p'}_{x_1,e} } \\
		&\ + R_{n+1}(F\calF,x_0)
	\end{align*}where \(m_{\gamma}=T^{\gamma}\), \(c'_r=\frac{(-1)^{r+1}}{(2r+1)!}\), $\mathcal G^{p'}_{x,e}$ was defined after \eqref{eq:8},
	and \(|R_{n+1}(F\calF,x_0)|\leq CT^{n+1-\epsilon}\) with \(C\) depending only on \(\epsilon, \gamma, p+p'\).
\end{lemma}

\begin{lemma}[Integration by part: gradient]
	\label{lem:int_by_part_grad}
	Let \(p'\geq 0\) be an integer. Let \(\calF\) be a monomial in \((\tilde{u}_{x}^k)_{x,k}\) of even degree at most \(2p'\) with coefficient \(1\). Let \(p:\Zd\to\Z_+\) be odd, and let \(q:\bbE_d\to\Z_+\) be even, both of finite support. Then, for any \(x_0\in\Zd,e'\in\rmB_+\), \(\alpha>n\geq 0\) integers, and \(\epsilon>0\),
	\begin{align*}
		&\lrangle{(\nabla_{x_0}^{e'}\phi) \phi^p (\nabla\phi)^{q} \calF} = \sum_{y\in \mathrm{supp}_p\cup \mathrm{supp}_q }  \nabla_{x_0}^{e'}G(\cdot,y)\lrangle{\partial_{y} (\phi^p (\nabla\phi)^{q}) \calF} +\\
		&\qquad + \mathds 1_{n\ge 1} \sum_{k=1}^{n} T^{k} \sum_{x_1,e} \nabla_{x_0,x_1}^{e',e} G^m \sum_{r+p' = k} c'_{r} \lrangle{\phi^p (\nabla\phi)^{q} (\nabla_{x_1}^{e}\phi)^{2r+1} \mathcal G^{p'}_{x_1,e} \calF} \\
		&\qquad + R_{n+1}'(p,q,\calF,x_0,e')
	\end{align*}where \(m=e^{-(\log T)^2}\),  and \(R_{n+1}'\) satisfies 
	\begin{equation*}
		|R_{n+1}'(p,q,\calF,x_0,e')|\leq CT^{\alpha} + CT^{n+1-\epsilon}\sum_{x\in\mathrm{supp}_p}\sum_{y\in\mathrm{supp}_q\cup\{x_0\}} \frac{\log^2(1+|x-y|)}{1+|x-y|}
	\end{equation*}
	with \(C\) depending only on \(\epsilon,\alpha, \normI{p}+\normI{q}+p',n\).
\end{lemma}

In both proofs, we will use Gaussian integration by parts 
with respect to a regularized version of the Gaussian measure described in Section~\ref{sec:perturbed_Gaussian_pic}, the regularization consisting in a mass term, which we add and subtract to the Gaussian weight. The first term in both lemmas is the standard Gaussian integration by parts, the second comes form the perturbation \(W_{L;h}\) (see Section~\ref{sec:perturbed_Gaussian_pic}) and the remainder term contains the effect of the regularization (mass) and of the perturbation \(\mathds{1}_{|\phi|\leq \sqrt{\beta}\pi}\) (see again Section~\ref{sec:perturbed_Gaussian_pic}). The proofs of the two lemmas will use the following a priori bounds on the decay of correlations, which will be proven in Appendix \ref{app:UB_corr_funct}.

\begin{lemma}
	\label{lem:setof_pt_to_setof_pt_decay}
	Let \(\epsilon>0,\alpha>0\), \(p'\geq 0\), \(p,q:\Zd\to\Z_+\) odd of finite support, and \(q':\Ed \to \bbZ_+\) odd of finite support.
	There exists \(C=C(\epsilon,\alpha,p'+\normI{p}+\normI{q})\) and \(C'=C(\epsilon,\alpha,p'+\normI{p}+\normI{q'})\) such that,
	 for any \(T<T_0\), and any \(\calF\) even monomial in \(\tilde{u}\) of degree at most \(2p'\),
	\begin{eqnarray}
	&		|\lrangle{\calF \phi^p\phi^{q} }|\leq CT^{\alpha} + C\beta^{\epsilon} \sum_{\substack{x:\ p_x\textnormal{ odd},\\ y:\ q_y\textnormal{ odd}}}\frac{\log(1+|x-y|)}{1+|x-y|},
	\label{lemma:eq1}
	\\
	&	|\lrangle{\calF \phi^p (\nabla\phi)^{q'}}|\leq C'T^{\alpha} + C'\beta^{\epsilon} \sum_{\substack{x:\ p_x\textnormal{ odd},\\ y\in\supp_{q'}}}\frac{\log(1+|x-y|)}{1+|x-y|}.\label{lemma:eq2}
	\end{eqnarray}
\end{lemma}

\begin{proof}[Proof of Lemma~\ref{lem:int_by_part_no_grad}]
Let us start with $L$ finite and $h>0$, and let us represent  $\langle\phi_{x_0} F\mathcal F \rangle_{L;h}$
as in \eqref{eq:12}. For any $0<m<1$, we rewrite the integral over $\phi$ in the numerator of \eqref{eq:12} (in the case of our observable of interest) as follows: 
$$  \int\limits_{[-\sqrt{\beta}\pi,\sqrt{\beta}\pi]^{\Lambda_L}}\!\!\!\!d\phi \, e^{-\frac12(\phi, (G_L^m)^{-1}\phi)+\frac{m^2}2\|\phi\|^2+W_{L;h}(\tilde u,\phi)}\phi_{x_0}F(\phi).$$
If we now integrate $\phi_{x_0}$ by parts with respect to the reference Gaussian weight $e^{-\frac12(\phi, (G_L^m)^{-1}\phi)}$, we obtain (denote \(\partial_y \equiv \frac{\partial}{\partial\phi_y}\)):
	\begin{align*}
		\lrangle{\phi_{x_0}F\calF}_{L;h} &= \sum_{y\in\mathrm{supp}_F}  G^m_L(x_0,y)\lrangle{\partial_{y}F\calF}_{L;h} + \sum_{z\in\Lambda_L} G^m_L(x_0,z)\lrangle{F \partial_{z}W_{L;h}\calF}_{L;h} \\
		&\quad + m^2\sum_{z\in\Lambda_L} G^m_L(x_0,z)\lrangle{\phi_z F\calF}_{L;h} 
		\\
		&\quad + \sum_{z\in\Lambda_L} G^m_L(x_0,z)\lrangle{(\delta(\phi_z+ \sqrt{\beta}\pi) - \delta(\phi_z- \sqrt{\beta}\pi)) F \calF}_{L;h}.
	\end{align*}
Now, taking \(L\to\infty\) for $h\in D^*$, and then $h\to 0^+$ along a sequence in $D^*$, all the terms in the 
two sides converge to their infinite volume, zero field, limits. Bounding the last term via the third item of Theorem~\ref{thm:moments_u_theta} and 
the Gaussian estimate~\eqref{eq:Gaussian_est:sum_of_G}, we obtain:
	\begin{align*}
		\lrangle{\phi_{x_0}F\calF} &= \sum_{y\in\mathrm{supp}_F}  G^m(x_0,y)\lrangle{\partial_{y}F\calF} + \sum_{z\in\mathbb Z^d} G^m(x_0,z)\lrangle{F \partial_{z}W\calF} \\
		&\quad + m^2\sum_{z\in\mathbb Z^d} G^m(x_0,z)\lrangle{\phi_z F\calF} + R_{\infty}(m,F\calF),
	\end{align*}
where in the second term in the right side we denoted by $W$ the formal $L\to\infty$, $h\to0$ limit of \eqref{eq:7}, and \(|R_{\infty}(m,F\calF)|\leq m^{-2}(C\beta)^{1/2+p+p'} e^{-c{\beta}}\) for some  \(c,C>0\) independent of $\beta$ and $m$ (as long as \(T<T_0\)). We denote the three other terms in the R.H.S. \(I,II,III\). Using~\eqref{eq:Gaussian_est:mass_removal_G} and the moment bounds 
in the first two items of Theorem~\ref{thm:moments_u_theta}, we find
	\begin{equation*}
		I = \sum_{y\in\mathrm{supp}_F}  G(x_0,y)\lrangle{\partial_{y}F\calF} + R_I(m,F\calF,x_0)
	\end{equation*} with \(|R_I(m,F\calF,x_0)|\leq C_{p+p'}m\) for some \(C_{p+p'}\) depending only on \(p+p'\). We then treat \(III\) when \(m=m_{\gamma}=T^{\gamma}\). Using Lemma~\ref{lem:setof_pt_to_setof_pt_decay}, one obtains that for any \(\epsilon>0\),
	\begin{equation*}
		|\lrangle{\phi_z F\calF}| \leq CT^{\gamma} + C\beta^{\epsilon/2}\sum_{y\in \mathrm{supp}_F} \frac{\log(1+|z-y|)}{1+|z-y|},
	\end{equation*}with \(C\) depending only on \(p+p',\epsilon,\gamma\). In particular, since \(\log(1+r)/(1+r)  \leq C_{\epsilon'}(1+r)^{\epsilon'-1}\) for any \(\epsilon'>0\), one has that, 
by choosing \(\epsilon'= \epsilon/(2\gamma)\) and using~\eqref{eq:Gaussian_est:sum_of_G} and~\eqref{eq:Gaussian_est:sum_of_G_gradG_with_power} (and remembering that we set \(m=T^{\gamma}\)),
	\begin{equation*}
		|III|\leq CT^{\gamma-\epsilon}
	\end{equation*}with \(C\) depending only on \(p+p',\epsilon,\gamma\). It remains to treat \(II\). We first compute
	\begin{align*}
		\sum_{z\in\mathbb Z^d} G^m(x_0,z) \partial_{z}W&= \sum_{z\in\mathbb Z^d} G^m(x_0,z) \partial_{z} \sum_{x\in\mathbb Z^d}\sum_{e}\Big[ \beta\rho_x\rho_{x+e}\cos(\sqrt{T}\nabla_x^{e}\phi) +\frac{1}{2}(\nabla_x^{e}\phi)^2\Big]\\
		&= \sum_{x\in\mathbb Z^d}\sum_{e} \nabla_x^{e}G^m(x_0,\cdot) \Big[ -\rho_x\rho_{x+e}\sqrt{\beta}\sin(\sqrt{T}\nabla_x^{e}\phi) +\nabla_x^{e}\phi\Big],
	\end{align*} where $\rho_x=\sqrt{1-T\|\tilde u_x\|^2}$. Now plugging in \(m=m_{\gamma} =T^{\gamma}\), and proceeding as in the proof of Lemma \ref{lem:tay} (i.e., in brief: Taylor expanding to order \(2\gamma-1\) both the \(\sin\) and the square root in the \(\rho_x\rho_{x+e}\) term, and using the bound on moments in Theorem~\ref{thm:moments_u_theta}
combined with Hölder's inequality), one finds that by the Gaussian estimate on the sum of gradients~\eqref{eq:Gaussian_est:sum_of_gradG},
	\begin{equation*}
		\Big|\sum_{k=1}^{2\gamma-1} T^{k} \sum_{x_1,\rme}\nabla_{x_1}^{\rme} G^{m_{\gamma}}(x_0,\cdot) \sum_{r+p'= k} c'_{r}\lrangle{F\calF (\nabla_{x_1}^{\rme}\phi)^{2r+1} \mathcal G^{p'}_{x_1,\rme}} - II \Big| \leq CT^{\gamma},
	\end{equation*}for some \(C\) depending only on \(\gamma+p+p'\), where \(c'_{r} =\frac{(-1)^{r+1}}{(2r+1)!}\), $\mathcal G^{p'}_{x,\rme}$ was defined after \eqref{eq:8}, and $\sum_{x_1}$ denotes summatin 
	over $\mathbb Z^d$. 
	Now,  one can use Lemma~\ref{lem:setof_pt_to_setof_pt_decay} combined with \eqref{eq:Gaussian_est:sum_of_gradG} and 
	\eqref{eq:Gaussian_est:sum_of_G_gradG_with_power} in the same fashion as before on the \(k\)s with \(n<k\leq 2\gamma-1\), to obtain that for any \(\epsilon>0\) there exists \(C\) depending on \(\epsilon\) and \(p+p'+\gamma\) only,
	\begin{equation*}
		\Big|\sum_{k=1}^{n} T^{k} \sum_{x_1,\rme}\nabla_{x_1}^{\rme} G^{m_{\gamma}}(x_0,\cdot) \sum_{r+p' = k} c'_{r}\lrangle{F\calF (\nabla_{x_1}^{\rme}\phi)^{2r+1} \mathcal G^{p'}_{x_1,\rme}} - II \Big| \leq CT^{n+1-\epsilon}.
	\end{equation*}Gathering all the above, and letting $m=m_\gamma=T^\gamma$ also in the error terms $R_\infty$ and $R_I$, gives the result.
\end{proof}

\begin{proof}[Proof of Lemma~\ref{lem:int_by_part_grad}]
	Write \(F_{p,q}= \phi^p (\nabla\phi)^{q}\). Starting as in the proof of Lemma~\ref{lem:int_by_part_no_grad}, for \(m=e^{-(\log T)^2}\),
	\begin{align*}
		\lrangle{(\nabla_{x_0}^{\rme'}\phi) F_{p,q} \calF} &= \sum_{y\in\mathrm{supp}_{F_{p,q}}}  \nabla_{x_0}^{\rme'}G^m(\cdot,y)\lrangle{\partial_{y}F_{p,q}\calF} + \sum_{z\in\mathbb Z^d} \nabla_{x_0}^{\rme'}G^m(\cdot,z)\lrangle{F_{p,q} \partial_{z}W\calF} \\
		&\quad + m^2\sum_{z\in\mathbb Z^d} \nabla_{x_0}^{\rme'}G^m(\cdot,z)\lrangle{\phi_z F_{p,q} \calF} + R_{\infty}(m,p,q,\calF,x_0)\\
		&=I+II+III+R_{\infty}(m,p,q,\calF,x_0),
	\end{align*}with \(|R_{\infty}(m,p,q,\calF,x_0)|\leq m^{-1} (C\beta)^{p'+\frac12(\normI{p}+\normI{q})}e^{-c{\beta} }\), for some \(C,c>0\) independent of $\beta$ and $m$. 
	Then, using~\eqref{eq:Gaussian_est:sum_of_gradG}, one has
	\begin{equation*}
		|III|\leq C_{p'+\normI{p}+\normI{q}} m.
	\end{equation*}Moreover, using~\eqref{eq:Gaussian_est:mass_removal_G} and Theorem~\ref{thm:moments_u_theta}, one has
	\begin{equation*}
	\big| I - \sum_{y\in\mathrm{supp}_{F_{p,q}}}  \nabla_{x_0}^{\rme'}G(\cdot,y)\lrangle{\partial_{y}F_{p,q}\calF} \big|\leq C_{p'+\normI{p}+\normI{q}} m.
	\end{equation*} To finish the proof of Lemma~\ref{lem:int_by_part_grad}, consider \(II\) and proceed as in the proof of Lemma~\ref{lem:int_by_part_no_grad} to obtain
	\begin{equation*}
		\Big|\sum_{k=1}^{\alpha} T^{k} \sum_{x_1,\rme}\nabla_{x_0,x_1}^{\rme',\rme} G^{m} \sum_{r+p' = k} c'_{r}\lrangle{F_{p,q}\calF (\nabla_{x_1}^{\rme}\phi)^{2r+1} \mathcal G^{p'}_{x_1,\rme}}- II \Big| \leq CT^{\alpha+1}(\log T)^2,
	\end{equation*}where we used~\eqref{eq:Gaussian_est:sum_of_gradgradG_general} and \(m=e^{-(\log T)^2}\), and \(C\) depends only on \(p'+\normI{p}+\normI{q}\) and \(\alpha\). We then apply Lemma~\ref{lem:setof_pt_to_setof_pt_decay} and~\eqref{eq:Gaussian_est:sum_of_gradgradG_with_power} to obtain that, for any \(r,p'\) with \(r+p'\leq \alpha\)
	and any $\epsilon>0$, 
	\begin{equation*}
		\label{eq:proof_of_lem_int_by_part_grad_eq1}
		\Big|\sum_{x_1,\rme}\nabla_{x_0,x_1}^{\rme',\rme} G^{m} \lrangle{F_{p,q}\calF (\nabla_{x_1}^{\rme}\phi)^{2r+1} \mathcal G^{p'}_{x_1,\rme} }\Big| \leq
 CT^{\alpha} + C\beta^{\epsilon} \sum_{x\in\mathrm{supp}_p}\sum_{y\in\mathrm{supp}_q\cup\{x_0\}} \frac{\log^2(1+|x-y|)}{1+|x-y|}
	\end{equation*}with \(C\) depending only on \(p'+\normI{p}+\normI{q}\), \(\alpha\), and \(\epsilon\), and $\sum_{x_1}$ denotes summation 
	over $\mathbb Z^d$.  This allows us to get rid of the terms with \(n+1\leq k\leq \alpha\). Combining all the previous estimates gives the lemma.
\end{proof}

\section{Inductive expansion}
\label{sec:Inductive_exp}
In this section we prove of our main results, Theorem \ref{thm:main} and Corollaries \ref{cor:magnet_second_order} and \ref{cor:2}, on the basis of an iterative application  of the integration-by-parts lemmas
stated and proved in Section \ref{sec:int_by_part} and of the moment bounds summarized in Theorem \ref{thm:moments_u_theta}. We start with the proof of Corollary \ref{cor:magnet_second_order}, which illustrates the general method in a simple, non-trivial, case. The general case is proved by induction, and discussed step-by-step in the following 
subsections. We recall that, whenever we write $\sum_{x_1}$, we always mean $\sum_{x_1\in\mathbb Z^d}$, and similarly for $\sum_{x_i}$, etc. As already mentioned above, we explicitly focus on the case $N>2$, the case $N=2$ being significantly simpler, in that the $u$ variables are not present; the interpretation and adaptation of the following discussion  
to the case $N=2$ is left to the reader. 

\subsection{Magnetization to second order}

We first present the proof of Corollary \ref{cor:magnet_second_order}. Shortcuts are taken as the full procedure will be described in the remaining of the section. First,
rewrite $S^N_0$ in terms of the coordinates $(\tilde u_0,\phi_0)$, so that 
$$\lrangle{S_0^N}=\lrangle{\sqrt{1-T\|\tilde u_0\|^2}\cos(\sqrt{T}\phi_0)}.$$
We expand the function under the average in the right hand side in Taylor series 
in $T$ up to second order and, applying Lemma \ref{lem:tay} (see also Remark \ref{rem:taylorr}), we obtain 
\begin{equation}\label{eq:6.1:1}
	\lrangle{S_0^N} = 1 -\frac{1}{2} T \lrangle{\phi_0^2} -\frac{1}{2} T \lrangle{\norm{\tilde{u}_0}^2} -\frac{1}{8}T^2 \lrangle{\norm{\tilde{u}_0}^4} + \frac{1}{4}T^2\lrangle{\norm{\tilde{u}_0}^2\phi_0^2}  + \frac{1}{24}T^2\lrangle{\phi_0^4} + O(T^3).
\end{equation}
Treat each term separately. First, using Lemma~\ref{lem:int_by_part_no_grad} with \(n=0\),
\begin{equation}\label{eq:6.1:2}\begin{split}
&	\lrangle{\phi_0^4} = 3G(0,0)\lrangle{\phi_0^2} + O(T^{1-\epsilon}) = 3G(0,0)^2 + O(T^{1-\epsilon}),\\
&	\lrangle{\norm{\tilde{u}_0}^2\phi_0^2} = G(0,0)\lrangle{\norm{\tilde{u}_0}^2} + O(T^{1-\epsilon}).
\end{split}\end{equation}
Then, using the leftover $O(N-1)$ symmetry of the model in the subspace orthogonal to the magnetic field, eq.\eqref{eq:remaining_symmetry_u_theta},  the moment bounds of Theorem~\ref{thm:moments_u_theta}, and applying Lemma~\ref{lem:int_by_part_no_grad} once more, we obtain
\begin{equation}\label{eq:6.1:3}\begin{split}
	\lrangle{\norm{\tilde{u}_0}^2} & = (N-2)\big(\lrangle{\phi_0^2} -\frac{T}{3}\lrangle{\phi_0^4} - T \lrangle{\norm{\tilde{u}_0}^2 \phi_0^2}\big) +O(T^2)\\
	& = (N-2)\big(\lrangle{\phi_0^2} -TG(0,0)^2 - T G(0,0)^2(N-2)\big) +O(T^{2-\epsilon}),
\end{split}\end{equation}
and
\begin{equation}\label{eq:6.1:4}\begin{split}
\lrangle{\norm{\tilde{u}_0}^4} & = (N-2)\lrangle{\phi_0^4} + (N-2)(N-3) \lrangle{(\tilde{u}_0^1)^2 \phi_0^2} + O(T)\\
	& =(N-2)3G(0,0)^2 + (N-2)(N-3) G(0,0)^2 + O(T^{1-\epsilon}).
\end{split}\end{equation}So far we used: (1) the leftover $O(N-1)$ symmetry in the directions orthogonal to the magnetic field, in order to generate \(\phi\)s when only \(\tilde{u}\)s remain; (2) Lemma~\ref{lem:int_by_part_no_grad}, but only in its ``trivial'' form (\(n=0\)). We then need to generate the expansion of \(\lrangle{\phi_0^2}\) to first order. This is where the controlled inductive procedure starts to show up. Using Lemma~\ref{lem:int_by_part_no_grad} with \(n=1\) (and \(\gamma=2\)), we obtain
\begin{equation}\begin{split}\label{def}
	 & \lrangle{\phi_0^2} = G(0,0) + O(T^{2-\epsilon})\\
	& \quad + T \sum_{x_1,e}\nabla_{x_1}^{\rme} G^{T^2}(0,\cdot)\Big(\frac{1}{3!} \lrangle{\phi_0 (\nabla_{x_1}^{e}\phi)^{3}} + \frac{1}{2} \lrangle{\phi_0 \nabla_{x_1}^{e}\phi (\norm{\tilde{u}_{x_1}}^{2}  +  \norm{\tilde{u}_{x_1+e}}^{2}) }\Big).
\end{split}\end{equation}One can then use Lemma~\ref{lem:int_by_part_grad} with \(n=0\) on the terms inside the sum to obtain
\begin{equation}\label{eq:35}\begin{split}
&	\lrangle{\phi_0 (\nabla_{x_1}^{e}\phi)^3} =  \nabla_{x_1}^{e}G(\cdot,0)\lrangle{ (\nabla_{x_1}^{e}\phi)^2} + 2\nabla^{e,e}_{x_1,x_1}G
\lrangle{\phi_0 \nabla_{x_1}^{e}\phi}+ R_{1}(x_1)\\
&	\lrangle{\phi_0 \nabla_{x_1}^{e}\phi ( \norm{\tilde{u}_{x_1}}^{2}+\norm{\tilde{u}_{x_1+e}}^{2})} =  \nabla_{x_1}^{e}G(\cdot,0)\big(\lrangle{ \norm{\tilde{u}_{x_1}}^{2}}
+\lrangle{\norm{\tilde{u}_{x_1+e}}^{2}}\big) + R_{2}(x_1)
\end{split}\end{equation}with \( R_{i}(x_1)\) satisfying, for $i=1,2$,  
\begin{equation}\label{eq:Ri}
	|R_{i}(x_1)|\leq CT^{3} + CT^{1-\epsilon} \frac{\log^2(1+|x_1|)}{1+|x_1|}.
\end{equation} Using again Lemmas~\ref{lem:int_by_part_no_grad} and \ref{lem:int_by_part_grad} 
with $n=0$ and the residual $O(N-1)$ symmetry in the directions orthogonal to the magnetic field, 
in \eqref{eq:35} we can rewrite
\begin{equation}\label{eq:36}\begin{split}
& \lrangle{ (\nabla_{x_1}^{e}\phi)^2} =\nabla^{e,e}_{x_1,x_1}G+O(T^{1-\epsilon}), \\
& \lrangle{\phi_0 \nabla_{x_1}^{e}\phi} = \nabla^{e}_{x_1}G(0,\cdot)+R_3(x_1),\\
& \lrangle{\norm{\tilde{u}_{x_1}}^{2}}= \lrangle{\norm{\tilde{u}_{x_1+e}}^{2}}=(N-2)G(0,0)+O(T^{1-\epsilon}),
\end{split}
\end{equation}
with $R_3(x_1)$ bounded as in \eqref{eq:Ri}. Plugging \eqref{eq:36} in \eqref{eq:35}, and then the resulting expressions in \eqref{def}, we find that, thanks to \eqref{eq:Gaussian_est:sum_of_gradG} and \eqref{eq:Gaussian_est:sum_of_G_gradG_with_power}, 
the contributions from the terms involving the error terms $R_i(x_1)$ and $O(T^{1-\epsilon})$ to the second line of \eqref{def} can all bounded by $CT^{2-\epsilon}$, for any $\epsilon>0$ and some $C>0$ depending on $\epsilon$, but independent of $T$. Therefore, we have 
\begin{equation}\begin{split}\label{def2}
	 & \lrangle{\phi_0^2} = G(0,0) + O(T^{2-\epsilon})\\
	& \quad + T \sum_{x_1,e}\nabla_{x_1}^{e} G^{T^2}(0,\cdot)\Big(\frac{1}{2} \nabla_{x_1}^{e}G(\cdot,0)\nabla^{e,e}_{x_1,x_1}G
	+(N-2) \nabla_{x_1}^{e}G(\cdot,0) G(0,0)\Big).
\end{split}\end{equation}
Finally, noting that, for any $\epsilon>0$,  $\big|\nabla_{x_1}^{e} G^{T^2}(0,\cdot)-\nabla_{x_1}^{e} G(0,\cdot)\big|\le C T^{2-2\frac{1+\epsilon}{d-1}}\frac1{(1+|x_1|)^{1+\epsilon}}$ 
(which follows from the second bound in Theorem \ref{thm:Gaussian_asymptotics} and from the use of Lemma \ref{lem:decay_and_mass_removal_implies_massDecay} 
with $m=T^2$, $\gamma=d-1$ and $\gamma'=1+\epsilon$), 
we can remove the mass \(T^2\) from the massive propagator up to an additional error of order $O(T^{2+\frac{d-3}{d-1}-\frac{2\epsilon}{d-1}})$. 
In conclusion, for any $\epsilon>0$,
\begin{equation*}\begin{split}
	\lrangle{S_0^N} & = 1 -\frac{T}2 (N-1)  G(0,0)
	+ \frac{T^2}{8}(N-1)(3N-5)G(0,0)^2 \\
	& - \frac{T^2}{2}(N-1)\sum_{x_1,e}\big(\nabla_{x_1}^{e} G(0,\cdot)\big)^2\Big(\frac{1}{2} \nabla_{x_1,x_1}^{e,e}G + (N-2)G(0,0)\Big)+O(T^{3-\epsilon}),
\end{split}\end{equation*}
which concludes the proof of Corollary \ref{cor:magnet_second_order}. 

\medskip

In the incoming subsections we will apply the same ideas used here in order to compute at a generic order in the low temperature expansion the correlation functions of the model. 
More precisely, we will inductively prove that, for any $p\in\mathfrak P$, $\tilde p\in\tilde{\mathfrak P}$, $n\ge 0$, $\epsilon>0$,
\begin{equation}\label{eq6.2.2}\lrangle{\phi^p\tilde u^{\tilde p}}=\lrangle{\phi^p\tilde u^{\tilde p}}^{(n)}+O(T^{n+1-\epsilon}),\end{equation}
where $\lrangle{\ }^{(n)}$ was defined in Definition \ref{def:3.1}. In view of Lemma \ref{lem:tay}, \eqref{eq6.2.2} implies our main result, Theorem \ref{thm:main}, with an explicit 
bound on the remainder and an explicit procedure for computing the coefficients $a_i$.

\subsection{Inductive hypotheses}

\begin{definition}[Inductive hypotheses \(\calH_{M,K}\)]
	We say that \(\calH_{M,K}\) holds if for any \(0<\epsilon<1\), and any integer \(\alpha> M\), the two following hold:
	\begin{itemize}
		\item[\(\calH_{M,K}^1\):] There exists \(C=C(\epsilon, K, M)\) such that: for any integer \(0\leq n \leq M\), \(p:\Zd\to \Z_+\) even, \(\tilde{p}:\Zd\times[N-2]\) even with \(\frac{\normI{\tilde{p}}+\normI{p}}{2}\leq K\),
		\begin{equation}
			\label{eq:induc_hyp:no_grad}
			\big| \lrangle{\phi^p \tilde{u}^{\tilde{p}}} - \lrangle{\phi^p \tilde{u}^{\tilde{p}}}^{(n)} \big|\leq CT^{n+1-\epsilon}.
		\end{equation}
		\item[\(\calH_{M,K}^2\):] There exists \(C'=C'(\epsilon,K,M,\alpha)\) and \(c_K\) such that: for any integer \(0\leq n \leq M\), even \(\tilde{p}:\Zd\times[N-2]\to \Z_+\), odd \(p:\Zd\to\Z_+\), and odd \(q:\bbE_d\to\Z_+\) with \(\frac{\normI{\tilde{p}}+\normI{p}+\normI{q}}{2}\leq K\),
		\begin{multline}
			\label{eq:induc_hyp:grad}
			\big| \lrangle{\phi^p \tilde{u}^{\tilde{p}} (\nabla\phi)^q} - \lrangle{\phi^p \tilde{u}^{\tilde{p}} (\nabla\phi)^q}^{(n)} \big| \leq\\
			\leq C'T^{\alpha} + C'T^{n+1-\epsilon}\sum_{x\in \supp_p}\sum_{y\in\supp_q} \frac{\log^{2+n}(1+|x-y|)}{1+|x-y|}.
		\end{multline}
	\end{itemize}
	We say that \(\calH_{M,\infty}\) holds if \(\calH_{M,K}\) holds for every \(K\geq 0\), and that \(\calH_{\infty,K}\) holds if \(\calH_{M,K}\) holds for every \(M\geq 0\).
\end{definition}

\begin{remark}
	By linearity, \(\calH_{M,K}\) implies a control over polynomials in \(\tilde{u},\phi,\nabla\phi\). The constants in the error terms depend then on the coefficients of the polynomial.
\end{remark}

\begin{remark}\label{Hinfty0}
	We automatically have the validity of \(\calH_{\infty,0}\).
\end{remark}

To lighten notation, we will not recall the domain and image of \(p,q,\tilde{p}\). The main result of this section is
\begin{theorem}
	\label{thm:main_induction}
	\(\calH_{M,K}\) holds for any \(M,K\geq 0\).
\end{theorem}
The proof of Theorem~\ref{thm:main_induction} will occupy the remainder of this section. Note that what we are in fact only interested in is the validity of \(\calH_{M,K}^1\). \(\calH_{M,K}^2\) is only needed for technical reasons (which will be highlighted when proving the induction step for \(\calH_{M,K}^1\)).

\subsection{Initiating the procedure}

\begin{lemma}
	\label{lem:induction:0_K_to_0_Kplus1}
	Suppose \(\calH_{0,K}\) holds. Then, \(\calH_{0,K+1}\) holds.
\end{lemma}
\begin{proof}
	Suppose \(\calH_{0,K}\) holds. Let \(0<\epsilon<1,\alpha>0\). We prove \(\calH_{0,K+1}^{2}\) and then \(\calH_{0,K+1}^{1}\).
	\subsubsection*{\(\calH_{0,K+1}^2\) part.} Let \(p\) be odd, \(q=\mathds{1}_{(x_0,e_0)} + \tilde{q}\) with \(\tilde{q}\) even, and \(\tilde{p}\) even with \(\frac{\normI{\tilde{p}}+\normI{p}+\normI{q}}{2}\leq K+1\). Using Lemma~\ref{lem:int_by_part_grad} with $n=0$, one obtains
	\begin{multline*}
		\big|\lrangle{\phi^p (\nabla\phi)^q \tilde{u}^{\tilde{p}}} -\sum_{y\in \supp_p \cup\supp_{\tilde{q}}}\nabla_{x_0}^{e_0} G(\cdot, y)\lrangle{\partial_{y} (\phi^p (\nabla\phi)^{\tilde{q}}) \tilde{u}^{\tilde{p}} } \big|\leq \\
		\leq CT^{\alpha} + CT^{1-\epsilon}\sum_{x\in\supp_p}\sum_{y\in\supp_q} \frac{\log^2(1+|x-y|)}{1+|x-y|}
	\end{multline*} for some \(C=C(\epsilon,\alpha,K+1)\). We can then use \(\calH_{0,K}^1\) to obtain
	\begin{equation*}
		|\lrangle{ (\nabla\phi)^{\tilde{q}} \partial_y \phi^p  \tilde{u}^{\tilde{p}} } - \lrangle{(\nabla\phi)^{\tilde{q}} \partial_y \phi^p  \tilde{u}^{\tilde{p}} }^{(0)}|\leq C' T^{1-\epsilon}
	\end{equation*}for some \(C'=C'(\epsilon,K)\). So,
	\begin{multline*}
		\big|\sum_{y\in \supp_p }\nabla_{x_0}^{e_0} G(\cdot, y)\big(\lrangle{(\nabla\phi)^{\tilde{q}} \partial_{y} \phi^p  \tilde{u}^{\tilde{p}} } - \lrangle{(\nabla\phi)^{\tilde{q}} \partial_{y} \phi^p  \tilde{u}^{\tilde{p}} }^{(0)} \big) \big|\leq \\
		\leq C'T^{1-\epsilon}\sum_{y\in \supp_p } \frac{c}{(1+|x_0-y|)^{d-1}},
	\end{multline*}where we used the decay of \(\nabla_{x_0}^eG(\cdot, y)\), see the second inequality in Theorem \ref{thm:Gaussian_asymptotics}. Now,
	\begin{equation*}
		\sum_{y}\nabla_{x_0}^{e_0} G(\cdot, y) \lrangle{ \phi^p \partial_{y}(\nabla\phi)^{\tilde{q}}  \tilde{u}^{\tilde{p}} } = \sum_{z,e}\tilde{q}_{z,e} \nabla_{x_0,z}^{e_0,e}G \lrangle{ \phi^p (\nabla\phi)^{\tilde{q}-\mathds{1}_{(z,e)}}  \tilde{u}^{\tilde{p}} }.
	\end{equation*}We can then use \(\calH_{0,K}^2\), the uniform boundedness of \(|\nabla_{x_0,z}^{e_0,e}G|\), and the previous estimates to obtain
	\begin{multline*}
		|\lrangle{\phi^p (\nabla\phi)^q \tilde{u}^{\tilde{p}}} -\sum_{y\in \supp_p \cup\supp_{\tilde{q}}}\nabla_{x_0}^{e_0} G(\cdot, y)\lrangle{\partial_{y} (\phi^p (\nabla\phi)^{\tilde{q}}) \tilde{u}^{\tilde{p}} }^{(0)} |\leq \\
		\leq C''T^{\alpha} + C''T^{1-\epsilon}\sum_{x\in\supp_p}\sum_{y\in\supp_q} \frac{\log^2(1+|x-y|)}{1+|x-y|},
	\end{multline*}with \(C''=C''(\epsilon,\alpha,K+1)\). Finally, using the extraction property~\eqref{eq:phi_contrib_extraction},
	\begin{align*}
		\lrangle{\phi^p (\nabla\phi)^q \tilde{u}^{\tilde{p}}}^{(0)} &= \Phi(\phi^p (\nabla\phi)^q) \lrangle{ \tilde{u}^{\tilde{p}}}^{(0)}\\
		&= \sum_y \nabla_{x_0}^{e_0} G(\cdot, y) \Phi\big(\partial_y(\phi^p (\nabla\phi)^{\tilde{q}})\big) \lrangle{ \tilde{u}^{\tilde{p}}}^{(0)}\\
		&= \sum_y \nabla_{x_0}^{e_0} G(\cdot, y) \lrangle{\partial_y(\phi^p (\nabla\phi)^{\tilde{q}}) \tilde{u}^{\tilde{p}}}^{(0)}
	\end{align*}where we used Gaussian integration by part in the second line. This concludes the proof of \(\calH_{0,K+1}^2\).

	\subsubsection*{\(\calH_{0,K+1}^1\) part.} Start with the case \(p\neq 0\). Then, \(p=\mathds{1}_{x_0} + p'\) with \(p'\) odd. Using Lemma~\ref{lem:int_by_part_no_grad} with $n=0$,
	\begin{equation*}
		\big| \lrangle{\phi^p \tilde{u}^{\tilde{p}}} - \sum_y G(x_0,y)\lrangle{\partial_y \phi^{p'} \tilde{u}^{\tilde{p}}} \big|\leq CT^{1-\epsilon}
	\end{equation*}with \(C=C(\epsilon,K+1)\). Using then \(\calH_{0,K}^1\) and the extraction property~\eqref{eq:phi_contrib_extraction} as before, one obtains the wanted claim. It remains to treat the \(p=0\) case. In that case, there exists \(k\in\{1,\cdots, N-2\}\) such that \(\tilde{p}^k\) is non-zero (and even). Denote 
\(\bar{p} = (\tilde p^1,\ldots,\tilde p^{k-1},0,\tilde p^{k+1},\ldots,\tilde p^{N-2})\). Using the residual $O(N-1)$ symmetry in the directions $\{1,\ldots,N-1\}$, eq.\eqref{eq:symmdef}, we get 
\begin{equation}
		\lrangle{ \tilde u^{\tilde p}} = \beta^{\|{\tilde{p}^k}\|/2 }\lrangle{\prod_{x}\Big(\sqrt{1-T\|\tilde u_x\|^2}\sin\big(\sqrt{T}\phi_x\big)\Big)^{\tilde p^k_x} \tilde{u}^{\bar{p}}}. 
		\end{equation}
Now, we apply Lemma \ref{lem:tay} with $n=\|{\tilde{p}^k}\|/2$ to the right hand side, thus getting 
$$\lrangle{ \tilde u^{\tilde p}} = \lrangle{\phi^{\tilde p^k} \tilde{u}^{\bar{p}}} + R_1(\tilde{p}),$$
with \(|R_1(\tilde{p})| \leq CT\), \(C=C(K+1)\). The rest is the same as in the \(p\neq 0\) case. This conclude the proof of \(\calH_{0,K+1}^1\).
\end{proof}

\begin{remark}\label{H0infty}
From Lemma \ref{lem:induction:0_K_to_0_Kplus1} and the validity of $\mathcal H_{0,0}$, one obtains the validity of \(\calH_{0,\infty}\).
\end{remark}

\subsection{Induction step}

To keep notation readable, we introduce the shorthands
\begin{equation*}
	f_i \equiv (\nabla_{x_i}^{e_i} \phi)^{2r_i +2},\quad g_i \equiv \calG_{x_i,e_i}^{p_i'},\quad f_A = \prod_{i\in A} f_i,\quad g_A = \prod_{i\in A} g_i.
\end{equation*}

We will use a simple identity on connected correlation of Gaussian:
\begin{claim}
	For any \(F\) monomial in \(\phi\) of odd degree,
	\begin{multline}
		\label{eq:trunc_grad_induction}
		\Phi(\phi_zF; f_1; \cdots ;f_k) = \sum_{y} G(z,y) \Phi(\partial_y F; f_1;\cdots; f_k ) + \\ +\sum_{i=1}^{k} (2r_i + 2) \nabla_{x_i}^{e_i}G(z,\cdot) \Phi\big(F (\nabla_{x_i}^{e_i}\phi)^{2r_i+1}; f_1;\cdots;f_{i-1};f_{i+1};\cdots; f_k \big).
	\end{multline}
\end{claim}
\begin{proof}
	It is a simple consequence of the sum-over-partitions formula \eqref{Urs2} for cumulants and of Gaussian integration by part. Letting $A_0$ be the element of 
$\pi\in\mathcal 	P(\{0,1,\ldots,k\})$ containing $0$, we have
	\begin{align*}
		&\Phi(\phi_zF; f_1; \cdots ;f_k) = \sum_{\pi\in\calP(\{0,\cdots, k\})} (|\pi|-1)! (-1)^{|\pi|-1} \Big(\prod_{B\in \pi, 0\notin B} \Phi(f_B)\Big) \Phi(\phi_zF f_{A_0\setminus 0})\\
		&\quad= \sum_{\pi\in\calP_k} (|\pi|-1)! (-1)^{|\pi|-1} \sum_{A\in \pi} \Big(\prod_{B\in \pi, B\neq A} \Phi(f_B)\Big) \big[ \Phi(\phi_zF f_{A}) - \Phi(\phi_zF )\Phi( f_{A}) \big].
\end{align*}
Integrating by parts $\phi_z$ in the two terms in brackets gives
			\begin{align*}
		&\Phi(\phi_zF; f_1; \cdots ;f_k) =\sum_{\pi\in\calP_k} (|\pi|-1)! (-1)^{|\pi|-1} \sum_{A\in \pi} \Big(\prod_{B\in \pi, B\neq A} \Phi(f_B)\Big)\cdot\\
		&\qquad \cdot \Big[\sum_{i\in A} (2r_i+2) \Phi(\phi_z\nabla_{x_i}^{e_i}\phi)\Phi\big(F (\nabla_{x_i}^{e_i}\phi )^{2r_i+1} f_{A\setminus i}\big)  +\sum_{y} G(z,y) \Phi(\partial_y F; f_{A}) \Big]\\
		&\quad= \sum_{i=1}^{k} (2r_i+2)\Phi(\phi_z\nabla_{x_i}^{e_i}\phi) \Phi\big( F (\nabla_{x_i}^{e_i}\phi)^{2r_i+1};f_1; \cdots; f_{i-1} ; f_{i+1} ;\cdots; f_k\big) +\\
		&\qquad+ \sum_{y} G(z,y) \Phi(\partial_y F; f_1; \cdots ;f_k),
	\end{align*}
	as desired.
\end{proof}

\begin{lemma}
	\label{lem:induction:Mplus1_K_and_M_inf_to_Mplus1_Kplus1}
	Suppose \(\calH_{M+1,K}\) and \(\calH_{M,\infty}\) hold. Then, \(\calH_{M+1,K+1}\) holds.
\end{lemma}
\begin{proof}
	Let \(\epsilon>0,\alpha>M+1\). We only have to consider \(n=M+1\). We again treat first \(\calH_{M+1,K+1}^2\) and then prove \(\calH_{M+1,K+1}^1\).
	
	\subsubsection*{\(\calH_{M+1,K+1}^2\) part.} As \(q\) is odd, there exist \(x_0\in\Zd\), \(e_0\in\rmB_+\) such that \(q= \mathds{1}_{(x_0,e_0)} + \tilde{q}\) with \(\tilde{q}\) even. Then, by Lemma~\ref{lem:int_by_part_grad},
	\begin{multline}
		\label{eq:proof_of_lem_induc_step_eq1}
		\Big| \lrangle{\nabla_{x_0}^{\rme_0} \phi \phi^p(\nabla\phi)^{\tilde{q}} \tilde{u}^{\tilde{p}} } - \sum_{y}\nabla_{x_0}^{e_0}G(\cdot, y) \lrangle{\partial_y (\phi^p(\nabla\phi)^{\tilde{q}}) \tilde{u}^{\tilde{p}}} \\
		- \sum_{k=1}^{M+1} T^k\sum_{x_1,e_1}\nabla_{x_0,x_1}^{e_0,e_1}G^{m} \sum_{r_1+p'_1 = k} c'_{r_1} \lrangle{(\nabla_{x_1}^{\rme_1}\phi)^{2r_1+1}\phi^p(\nabla\phi)^{\tilde{q}} \tilde{u}^{\tilde{p}} g_1} \Big|\leq\\
		\leq CT^{\alpha} + CT^{M+2-\epsilon}\sum_{x\in\supp_p}\sum_{y\in\supp_q} \frac{\log^2(1+|x-y|)}{1+|x-y|},
	\end{multline}with \(C= C(\epsilon,\alpha, K+1, M+1)\), and \(m=e^{-(\log T)^2}\). We can then use \(\calH_{M+1,K}^1\) to obtain
	\begin{equation*}
		\big|\lrangle{(\nabla\phi)^{\tilde{q}} \partial_y \phi^p \tilde{u}^{\tilde{p}}} - \lrangle{(\nabla\phi)^{\tilde{q}} \partial_y \phi^p \tilde{u}^{\tilde{p}}}^{(M+1)}\big|\leq C' T^{M+2-\epsilon},
	\end{equation*}with \(C'=C'(\epsilon, K, M+1)\). So, using the decay of \(|\nabla_{x_0}^{e_0}G(\cdot, y)|\) (see the second bound of Theorem \ref{thm:Gaussian_asymptotics}),
	\begin{equation*}
		\Big|\sum_{y}\nabla_{x_0}^{e_0}G(\cdot, y) \big[\lrangle{(\nabla\phi)^{\tilde{q}} \partial_y \phi^p \tilde{u}^{\tilde{p}}}- \lrangle{(\nabla\phi)^{\tilde{q}} \partial_y \phi^p \tilde{u}^{\tilde{p}}}^{(M+1)}\big]\Big|\leq C'T^{M+2-\epsilon}\sum_{y\in \supp_{p}}\frac{c_d}{(1+|x_0-y|)^{d-1}}
	\end{equation*}with \(c_d\) depending only on \(d\). Next,
	\begin{equation*}
		\sum_{y}\nabla_{x_0}^{e_0}G(\cdot, y) \lrangle{ \phi^p \partial_y(\nabla\phi)^{\tilde{q}} \tilde{u}^{\tilde{p}}} = \sum_{z,e}\tilde{q}_{z,e}\nabla_{x_0,z}^{e_0,e}G \lrangle{ \phi^p (\nabla\phi)^{\tilde{q}-\mathds{1}_{(z,\rme)}} \tilde{u}^{\tilde{p}}}.
	\end{equation*}So, by the uniform boundedness of \(|\nabla_{x_0,z}^{e_0,e}G|\), and \(\calH_{M+1,K}^2\),
	\begin{multline*}
		\Big|\sum_{y}\nabla_{x_0}^{e_0}G(\cdot, y) \big[\lrangle{ \phi^p \partial_y(\nabla\phi)^{\tilde{q}} \tilde{u}^{\tilde{p}}}-  \lrangle{ \phi^p \partial_y(\nabla\phi)^{\tilde{q}} \tilde{u}^{\tilde{p}}}^{(M+1)}\big] \Big|\leq\\
		\leq C''T^{\alpha} + C''T^{M+2-\epsilon} \sum_{y\in \supp_{\tilde{q}}} \sum_{x\in \supp_p} \frac{\log^{M+3}(1+|x-y|)}{1+|x-y|}
	\end{multline*}with \(C''= C''(\epsilon,\alpha,K+1,M+1)\). Finally, using \(\calH_{M,\infty}^2\), there exists \(C'''=C'''(\epsilon,\alpha,K+1,M+1)\) such that for \(1\leq k\leq M\),
	\begin{multline*}
		\big|\lrangle{(\nabla_{x_1}^{e_1}\phi)^{2r_1+1}\phi^p(\nabla\phi)^{\tilde{q}} \tilde{u}^{\tilde{p}} g_1} - \lrangle{(\nabla_{x_1}^{e_1}\phi)^{2r_1+1}\phi^p(\nabla\phi)^{\tilde{q}} \tilde{u}^{\tilde{p}} g_1}^{(M+1-k)}\big|\leq \\
		\leq C''' T^{\alpha} + C'''T^{M+2-k-\epsilon} \sum_{x\in\supp_p}\sum_{y\in \supp_{\tilde{q}} \cup \{x_1\}} \frac{\log^{M+3-k}(1+|x-y|)}{1+|x-y|}.
	\end{multline*}Plugging all the estimates in~\eqref{eq:proof_of_lem_induc_step_eq1}, using triangular inequality, and using Theorem~\ref{thm:Gaussian_estimates}, more precisely~\eqref{eq:Gaussian_est:sum_of_gradgradG_general} and~\eqref{eq:Gaussian_est:sum_of_gradgradG_with_power}, we obtain
	\begin{equation}\label{llong}\begin{split}  
&		\Big| \lrangle{\nabla_{x_0}^{e_0} \phi \phi^p(\nabla\phi)^{\tilde{q}} \tilde{u}^{\tilde{p}} } - \sum_{y}\nabla_{x_0}^{e_0}G(\cdot, y) \lrangle{\partial_y (\phi^p(\nabla\phi)^{\tilde{q}}) \tilde{u}^{\tilde{p}}}^{(M+1)} \\
&		- \sum_{k=1}^{M+1} T^k\sum_{x_1,e_1}\nabla_{x_0,x_1}^{e_0,e_1}G^{m} \sum_{r_1+p'_1 = k} c'_{r_1} \lrangle{(\nabla_{x_1}^{e_1}\phi)^{2r_1+1}\phi^p(\nabla\phi)^{\tilde{q}} \tilde{u}^{\tilde{p}} g_1}^{(M+1-k)} \Big|\leq\\
&		\leq \tilde{C}T^{\alpha} + \tilde{C}T^{M+2-\epsilon}(\log T)^2\sum_{x\in\supp_p}\sum_{y\in\supp_q} \frac{\log^{M+3}(1+|x-y|)}{1+|x-y|},\end{split}
	\end{equation}with \(\tilde{C}=\tilde{C}(\epsilon,\alpha,K+1,M+1)\). If we now use the bound 
$$|\nabla_{x_0,x_1}^{e_0,e_1}G - \nabla_{x_0,x_1}^{e_0,e_1}G^{m}|\leq Cm^{1/(2d)}(1+|x_1-x_0|)^{-d+1/2}$$ following from Theorems~\ref{thm:Gaussian_estimates} and~\ref{thm:Gaussian_asymptotics} and Lemma~\ref{lem:decay_and_mass_removal_implies_massDecay}, the bound
$$|\lrangle{(\nabla_{x_1}^{e_1}\phi)^{2r_1+1}\phi^p(\nabla\phi)^{\tilde{q}} \tilde{u}^{\tilde{p}} g_1}^{(M+1-k)}|\le \sum_{y\in\supp_p\, \cup\, \supp_{\tilde q}}(1+|x_1-y|)^{-d+1+\epsilon}$$
following from Lemma~\ref{lem:Taylor_n_a_priori_decay_grad}, and the summability result of Lemma~\ref{app:lem:sum_Zd_a_b_to_log}, we see that in the second line of \eqref{llong} we can replace the massive propagator \(\nabla_{x_0,x_1}^{e_0,e_1}G^{m}\) by \(\nabla_{x_0,x_1}^{e_0,e_1}G\) up to an error term \(\leq C T m^{1/(2d)}\), which can be re-absorbed in $\tilde C T^\alpha$ up to a redefinition of $\tilde C$. It remains to prove that 
	\begin{multline}
		\label{eq:proof_induction_step_eq2}
\sum_{y}\nabla_{x_0}^{e_0}G(\cdot, y) \lrangle{\partial_y (\phi^p(\nabla\phi)^{\tilde{q}}) \tilde{u}^{\tilde{p}}}^{(M+1)} + \sum_{k=1}^{M+1} T^{k}\sum_{x_1,e_1}\nabla_{x_0,x_1}^{e_0,e_1}G\ \cdot\\
\cdot \sum_{r_1+p'_1 = k} c'_{r_1} \lrangle{(\nabla_{x_1}^{e_1}\phi)^{2r_1+1}\phi^p(\nabla\phi)^{\tilde{q}} \tilde{u}^{\tilde{p}} g_1}^{(M+1-k)} \equiv \lrangle{\nabla_{x_0}^{e_0} \phi\, \phi^p(\nabla\phi)^{\tilde{q}} \tilde{u}^{\tilde{p}} }^{(M+1)} .
	\end{multline}Setting \(F= \phi^p(\nabla\phi)^{\tilde{q}}\), we can now use the extraction formula~\eqref{eq:phi_contrib_extraction} to obtain:
	\begin{multline*}
		\lrangle{\partial_y F \tilde{u}^{\tilde{p}}}^{(M+1)} = \Phi(\partial_y F )\lrangle{\tilde{u}^{\tilde{p}}}^{(M+1)}+\\
		+ \sum_{k=1}^{M+1}\frac{1}{k!} \sum_{s=k}^{M+1} T^{s}  \sum_{\substack{s_1,\cdots,s_k \geq 1\\ \sum s_l = s}} \sum_{\substack{x_1,\cdots, x_k\\ e_1,\cdots,e_k}}\ \sum_{\substack{r_1,p'_1,\cdots, r_k,p'_k\ge 0 \\ p'_l+r_l = s_l} }\big(\prod_{l=1}^k c_{r_l}\big)\Phi\big( \partial_yF ; f_1;\cdots; f_k\big)\lrangle{\tilde{u}^{\tilde{p}} \prod_{l=1}^{k}g_l}^{(M+1-s)},
	\end{multline*}where $c_r=\frac{(-1)^{r+1}}{(2r+2)!}=\frac1{2r+2}c_r'$. Similarly, renaming $k$ as $s_1$ in the second line of \eqref{eq:proof_induction_step_eq2}, 
for \(s_1=1,\cdots M+1\),
	\begin{multline*}
		\lrangle{(\nabla_{x_1}^{e_1}\phi)^{2r_1+1}F \tilde{u}^{\tilde{p}} g_1}^{(M+1-s_1)} = \Phi\big((\nabla_{x_1}^{e_1}\phi)^{2r_1+1}F\big) \lrangle{ \tilde{u}^{\tilde{p}} g_1}^{(M+1-s_1)}	\\
		+ \mathds{1}_{M+1>s_1}\sum_{k=1}^{M+1-s_1}\frac{1}{k!} \sum_{s=k}^{M+1-s_1} T^{s} \sum_{\substack{s_2,\cdots,s_{k+1} \geq 1\\ \sum s_l = s}} \sum_{\substack{x_2,\cdots, x_{k+1}\\ e_2,\cdots,e_{k+1}}}\ \sum_{\substack{r_2,p'_2,\cdots, r_{k+1},p'_{k+1}\ge 0\\ p'_l+r_l = s_l} }\cdot \\ \cdot\big(\prod_{l=2}^{k+1} c_{r_l}\big)\Phi\big( (\nabla_{x_1}^{e_1}\phi)^{2r_1+1}F ; f_2;\cdots; f_{k+1}\big)\lrangle{\tilde{u}^{\tilde{p}} \prod_{l=1}^{k+1}g_l}^{(M+1-s_1-s)}.
	\end{multline*}Plugging these in the L.H.S. of~\eqref{eq:proof_induction_step_eq2}, and using the integration by parts formula for the Gaussian measure $\Phi$, one obtains
	\begin{align*}
		&\Phi(\nabla_{x_0}^{e_0}\phi F )\lrangle{\tilde{u}^{\tilde{p}}}^{(M+1)} + \sum_{k=1}^{M+1}\frac{1}{k!} \sum_{s=k}^{M+1} T^{s} \sum_{\substack{s_1,\cdots,s_k \geq 1\\ \sum s_l = s}} \sum_{\substack{x_1,\cdots, x_k\\ e_1,\cdots,e_k}}\ \sum_{\substack{r_1,p'_1,\cdots, r_k,p'_k\ge 0\\ p'_l+r_l = s_l} }\ \big(\prod_{l=1}^k c_{r_l}\big)\times\\
		& \hskip3.5truecm\times \sum_{y}\nabla_{x_0}^{e_0}G(\cdot, y) \Phi\big( \partial_yF ; f_1;\cdots; f_k\big)\lrangle{\tilde{u}^{\tilde{p}} \prod_{l=1}^{k}g_l}^{(M+1-s)} \\
		& + \sum_{s_1=1}^{M+1} T^{s_1} \sum_{x_1,e_1}  \sum_{r_1+p'_1 = s_1} c_{r_1}(2r_1+2) \nabla_{x_0,x_1}^{e_0,e_1}G\,\Phi\big((\nabla_{x_1}^{\rme_1}\phi)^{2r_1+1}F\big) \lrangle{ \tilde{u}^{\tilde{p}} g_1}^{(M+1-s_1)} \\
		& + \sum_{k=2}^{M+1} \frac{1}{k!} \sum_{s=k}^{M+1} T^{s} \sum_{\substack{s_1,\cdots,s_{k} \geq 1\\ \sum s_l = s}} \sum_{\substack{x_1,\cdots, x_{k}\\ e_1,\cdots,e_{k}}}\  \sum_{\substack{r_1,p'_1,\cdots, r_{k},p'_{k}\ge 0\\ p'_l+r_l = s_l} } \sum_{i=1}^{k} (2r_i+2) \times\\
		&\quad  \times \big(\prod_{l=1}^{k} c_{r_l}\big) \nabla_{x_0,x_i}^{e_0,e_i}G\, \Phi\big( (\nabla_{x_i}^{e_i}\phi)^{2r_i+1}F ; f_1;\cdots;f_{i-1};f_{i+1};\cdots; f_{k}\big)\lrangle{\tilde{u}^{\tilde{p}} \prod_{l=1}^{k}g_l}^{(M+1-s)},
	\end{align*}where we symmetrized the role of index \(1\) in the last line. Re-grouping the terms with the same values of \(k,s\), we can rewrite this as
	\begin{multline*}
		\Phi(\nabla_{x_0}^{e_0}\phi F )\lrangle{\tilde{u}^{\tilde{p}}}^{(M+1)}\ +\ \sum_{k=1}^{M+1}\frac{1}{k!} \sum_{s=k}^{M+1} T^{s} \sum_{\substack{s_1,\cdots,s_{k} \geq 1\\ \sum s_l = s}} \sum_{\substack{x_1,\cdots, x_{k}\\ e_1,\cdots,e_{k}}} \, \sum_{\substack{r_1,p'_1,\cdots, r_{k},p'_{k}\ge 0\\ p'_l+r_l = s_l} } \lrangle{\tilde{u}^{\tilde{p}} \prod_{l=1}^{k}g_l}^{(M+1-s)}\times\\
		\times \big(\prod_{l=1}^{k} c_{r_l}\big) \Big[ \sum_{i=1}^{k} (2r_i+2) \nabla_{x_0,x_i}^{e_0,e_i}G\, \Phi\big( (\nabla_{x_i}^{e_i}\phi)^{2r_i+1}F ; f_1;\cdots;f_{i-1};f_{i+1};\cdots; f_{k}\big)  +\\+ \sum_{y}\nabla_{x_0}^{e_0}G(\cdot, y) \Phi\big( \partial_yF ; f_1;\cdots; f_k\big)\Big].
	\end{multline*}But now, it follows from~\eqref{eq:trunc_grad_induction} that the term in brackets is simply \(\Phi\big(\nabla_{x_0}^{e_0}\phi F; f_1;\cdots; f_k\big)\), which implies~\eqref{eq:proof_induction_step_eq2} (by another look at the extraction formula~\eqref{eq:phi_contrib_extraction}).

	\subsubsection*{\(\calH_{M+1,K+1}^1\) part.} We first consider \(p\neq 0\). We can then write \(p=\mathds{1}_{x_0} + \bar{p}\) with \(\bar{p}\) odd, \(x_0\in\Zd\). Using Lemma~\ref{lem:int_by_part_no_grad} with \(\gamma = 2(d-1)(M+1)\), we have
	\begin{multline*}
		\Big|\lrangle{\phi^p \tilde{u}^{\tilde{p}}} - \sum_{y} G(x_0,y)\lrangle{\partial_y\phi^{\bar{p}} \tilde{u}^{\tilde{p}}} \\
		- \sum_{k=1}^{M+1} T^k \sum_{x_1,e_1} \nabla_{x_1}^{e_1}G^{m_{\gamma}}(x_0,\cdot) \sum_{r_1+p_1'= k} c'_{r_1}\lrangle{(\nabla_{x_1}^{e_1}\phi)^{2r_1+1} \phi^{\bar{p}}\tilde{u}^{\tilde{p}} g_1} \Big|\leq CT^{M+2-\epsilon}
	\end{multline*}with \(C= C(K+1,M+1,\epsilon)\), and \(m_{\gamma} =T^\gamma  = T^{2(d-1)(M+1)}\). We can then use \(\calH_{M,\infty}^2\)  to obtain that for \(k=1,\cdots,n\),
	\begin{multline*}
		\big|\lrangle{(\nabla_{x_1}^{e_1}\phi)^{2r_1+1} \phi^{\bar{p}}\tilde{u}^{\tilde{p}} g_1} - \lrangle{(\nabla_{x_1}^{e_1}\phi)^{2r_1+1} \phi^{\bar{p}}\tilde{u}^{\tilde{p}} g_1}^{(M+1-k)}\big|\leq \\
		\leq C'T^{\gamma+ M +1} + T^{M+2-k-\epsilon/2}C' \sum_{y\in \supp_{\bar{p}}} \frac{\log^{M+3-k}(1+|y-x_1|)}{1+|y - x_1|},
	\end{multline*}with \(C'=C'(K+1,M+1,\epsilon)\), and use \(\calH_{M+1,K}^1\) to obtain
	\begin{equation*}
		\big|\lrangle{\partial_y\phi^{\bar{p}} \tilde{u}^{\tilde{p}}} - \lrangle{\partial_y\phi^{\bar{p}} \tilde{u}^{\tilde{p}}}^{(M+1)} \big|\leq C'' T^{M+2-\epsilon},
	\end{equation*}with \(C''=C''(K+1,M+1,\epsilon)\). Using these and~\eqref{eq:Gaussian_est:sum_of_gradG},~\eqref{eq:Gaussian_est:sum_of_G_gradG_with_power}, we get
	\begin{multline*}
		\Big|\lrangle{\phi^p \tilde{u}^{\tilde{p}}} - \sum_{y} G(x_0,y)\lrangle{\partial_y\phi^{\bar{p}} \tilde{u}^{\tilde{p}}}^{(M+1)} \\
		- \sum_{k=1}^{M+1} T^k \sum_{x_1,\rme_1} \nabla_{x_1}^{e_1}G^{m_{\gamma}}(x_0,\cdot) \sum_{r_1+p_1'= k} c'_{r_1}\lrangle{(\nabla_{x_1}^{e_1}\phi)^{2r_1+1} \phi^{\bar{p}}\tilde{u}^{\tilde{p}} g_1}^{(M+1-k)} \Big|\leq C'''T^{M+2-\epsilon},
	\end{multline*}with \(C'''=C'''(K+1,M+1,\epsilon)\). We can then use the decay $$|\lrangle{(\nabla_{x_1}^{e_1}\phi)^{2r_1+1} \phi^{\bar{p}}\tilde{u}^{\tilde{p}} g_1}^{(M+1-k)}|\leq \sum_{y\in\supp_{\bar{p}}}\frac{c}{(1+|x_1-y|)^{d-5/4}}$$ from Lemma~\ref{lem:Taylor_n_a_priori_decay_grad}, combined with $$|\nabla_{x_1}^{e_1}G^{m_{\gamma}}(x_0,\cdot) - \nabla_{x_1}^{e_1}G(x_0,\cdot)|\leq cm_{\gamma}^{1/(2d-2)} (1+|x_0-x_1|)^{-d+3/2}$$ (which follows from~\eqref{eq:Gaussian_est:mass_removal_G}, the decay bound of Theorem~\ref{thm:Gaussian_asymptotics}, and Lemma~\ref{lem:decay_and_mass_removal_implies_massDecay}), and with Lemma~\ref{app:lem:sum_Zd_a_b_to_log}, 
 to replace the massive propagator \(\nabla_{x_1}^{e_1}G^{m_{\gamma}}\) by its non-massive version, \(\nabla_{x_1}^{e_1}G\), up to an error smaller than \( c'Tm_{\gamma}^{1/(2d-2)} = c'T^{M+2}\). The equality
	\begin{multline*}
		 \sum_{y} G(x_0,y)\lrangle{\partial_y\phi^{\bar{p}} \tilde{u}^{\tilde{p}}}^{(M+1)} \\
		+ \sum_{k=1}^{M+1} T^k \sum_{x_1,\rme_1} \nabla_{x_1}^{e_1}G(x_0,\cdot) \sum_{r_1+p_1'= k} c_{r_1}'\lrangle{(\nabla_{x_1}^{e_1}\phi)^{2r_1+1} \phi^{\bar{p}}\tilde{u}^{\tilde{p}} g_1}^{(M+1-k)}\equiv \lrangle{\phi^p \tilde{u}^{\tilde{p}}}^{(M+1)} 
	\end{multline*}follows exactly as in the proof of \(\calH_{M+1,K+1}^{2}\), and concludes the proof of \(\calH_{K+1,M+1}^1\) when \(p\neq 0\). We now turn to the \(p = 0\) case. Without loss of generality, we can assume that \(\tilde{p}^1\) is non-zero and even. Let \(\tilde{q}_x^k = \tilde{p}_x^k\) for \(k=2,\cdots,N-2\) and \(\tilde{q}_x^1 =0\). We then use the remaining \(O(N-1)\) symmetry through~\eqref{eq:remaining_symmetry_u_theta} to obtain
	\begin{equation*}
		\lrangle{\tilde{u}^{\tilde{p}}} = \lrangle{(\tilde{u}^1)^{\tilde{p}^1} \tilde{u}^{\tilde{q}}} = \beta^{\|{\tilde{p}^1}\|/2}\lrangle{\prod_{x} \big(\sqrt{1-T\norm{\tilde{u}_x}^2}\sin\big(\sqrt{T}\phi_x\big) \big)^{\tilde{p}_x^1} \tilde{u}^{\tilde{q}}}.
	\end{equation*}
We now apply Lemma \ref{lem:tay}: letting $\phi^{\tilde p^1}(1+\sum_{s\ge 1}T^s \sum_{n,\tilde n} a^{n,\tilde n}_{\tilde p^1, s} \phi^{n}\tilde u^{\tilde n})$ be the formal low temperature expansion of 
$\beta^{\|{\tilde{p}^1}\|/2 }\prod_{x}\Big(\sqrt{1-T\|\tilde u_x\|^2}\sin\big(\sqrt{T}\phi_x\big)\Big)^{\tilde p^1_x} $ obtained via the first two of \eqref{exptutti} (cf. with \eqref{expAalpha}), 
where the sum over $n,\tilde n$ runs over even tuples in $\mathfrak P,\tilde{\mathfrak P}$, and the coefficients $a^{n,\tilde n}_{\tilde p^1, s}$
are non zero only if $\|n\|_1+\|\tilde n\|_1=2s$, we obtain 
	\begin{equation*}
		\Big|\lrangle{\tilde{u}^{\tilde{p}}} -\lrangle{\phi^{\tilde p^1}\tilde u^{\tilde q}}-\sum_{s=1}^{M+1} T^s \sum_{n,\tilde n} a^{n,\tilde n}_{\tilde p^1, s}  \lrangle{\phi^{n+\tilde p^1}\tilde u^{\tilde n+\tilde q}} \Big| \leq CT^{M+2},
	\end{equation*}with \(C= C(K,M)\). 
	The claim then reduces to the \(p\neq 0\) case and to \(\calH^1_{M,\infty}\). This concludes the proof of the validity of \(\calH_{M+1,K+1}^1\) and therefore of the Lemma.
\end{proof}

From Lemma \ref{lem:induction:Mplus1_K_and_M_inf_to_Mplus1_Kplus1}, the validity of $\mathcal H_{1,0}$ (see Remark \ref{Hinfty0}) and of $\mathcal H_{0,\infty}$ (see Remark \ref{H0infty}), the conclusion of Theorem \ref{thm:main_induction} follows. In particular, this implies that \eqref{eq6.2.2} holds, for any $n\ge 0$ and $\epsilon>0$, and, as remarked 
after \eqref{eq6.2.2}, this also implies our main result, Theorem \ref{thm:main}. 

\subsection{Proof of Corollary \ref{cor:2}}

\begin{proof} 
It is sufficient to check convergence for finite degree polynomials in  \(\sqrt{\beta}(S^1,S^2,$ $\ldots,S^{N-1})\), as they are dense in the continuous local functions; and, by linearity, we can further reduce to finite degree monomials. Therefore, using the change of variables from $S$ to $(\phi,\tilde u)$, one needs to control 
$$\lrangle{ \prod_{x}\Big( \big(\sqrt{\beta}\sqrt{1-T\|\tilde u_x\|^2}\sin\big(\sqrt{T}\phi_x\big)\big)^{\alpha^{N-1}_x} \prod_{k=1}^{N-2}\tilde{u}_x^{\alpha^k_x}\Big)},$$ 
for some $\alpha^1,\ldots,\alpha^{N-1}:\mathbb Z^d\to\mathbb Z_+$ of finite support. By Lemma \ref{lem:tay},  this is equal to \(\lrangle{ \prod_{x}\Big( \phi_x^{\alpha^{N-1}_x} \prod_{k=1}^{N-2}\tilde{u}_x^{\alpha^k_x}\Big)} +O(T)$. By Theorem \ref{thm:main_induction} this is, in turn, equal to $$\lrangle{ \prod_{x}\Big( \phi_x^{\alpha^{N-1}_x} \prod_{k=1}^{N-2}\tilde{u}_x^{\alpha^k_x}\Big)}^{(0)} +O(T^{1-\epsilon})= \Phi(\phi^{\alpha^{N-1}}) \lrangle{ \prod_{x}\prod_{k=1}^{N-2}\tilde{u}_x^{\alpha^k_x}}^{(0)}+ O(T^{1-\epsilon}),$$ for any $\epsilon>0$. 
Using the residual $O(N-1)$ symmetry \eqref{eq:remaining_symmetry_u_theta} and iterating \((N-2)\) times gives the result.
\end{proof}

\section*{Acknowledgements}
We gratefully acknowledges financial support of the European 
Research Council through the ERC CoG UniCoSM, grant agreement n.\,724939. A.G. also acknowledges support from MIUR, through the PRIN 2017 project MaQuMA, PRIN201719VMAST01 and thanks GNFM-INdAM Gruppo Nazionale per la Fisica Matematica.
S.O. is supported by the Swiss NSF grant 200021\_182237 and is a member of the NCCR SwissMAP. Most of this work was completed while S.O. was supported by an Swiss NSF early PostDoc.Mobility Grant. He thanks the university Roma Tre for its hospitality.

\appendix

\section{Proof of Proposition \ref{thm:Infrared_infinite_volume_magn_bound}}\label{prova:IB}

Recall that 
\begin{equation*}
	\psi(v) \equiv \psi_{\beta}(v) = \lim_{L\to \infty} \frac{1}{|\Lambda_L|}\log Z_{L;\beta}\big( e^{\sum_{x\in \Lambda_L} v\cdot S_x }\big).
\end{equation*}
By \(O(N)\) invariance, \(\psi(Rv) = \psi(v)\) for any matrix \(R\in O(N)\); moreover, if $\psi$ is differentiable at $hs$, for some $h>0$ and $s\in\mathbb S^{N-1}$, then 
$\psi$ is differentiable at $hs'$, for any $s'\in\mathbb S^{N-1}$. 
Recall also that \(v\mapsto \psi(v)\) is a convex function from \(\R^N\) to \(\R\). Therefore, \(\psi\) is differentiable almost everywhere. Define
\begin{equation*}
	D^* = \{h\in [0,\infty):\  \psi \textnormal{ is differentiable at } hs,\ \forall s\in \bbS^{N-1}\},
\end{equation*}
which is (at least) dense in \([0,\infty)\). For \(h\in D^*\), let \(\rmJ_{h\hat\rme_N}:\R^N\to \R\) denote the derivative (gradient) of \(\psi\) at \(h\hat\rme_N\). Since \(\psi\) is the limit of a sequence of convex differentiable functions, one has that for every \(s\in \bbS^{N-1}\),
	\begin{equation*}
		 \rmJ_{h\hat\rme_N}(s) = \lim_{L\to \infty} \frac{1}{|\Lambda_L|} \frac{Z_{L;\beta}\big( \sum_{x\in \Lambda_L}  S_x\cdot s
		e^{h\sum_{x\in \Lambda_L}  S_x^N }\big)}{Z_{L;\beta}\big( e^{h\sum_{x\in \Lambda_L}  S_x^N }\big)} = \lim_{L\to \infty} \lrangle{S_0\cdot s}_{L;\beta,h}.
	\end{equation*}In particular, by the residual $O(N-1)$ symmetry in the directions orthogonal to $\hat\rme_N$, $$\rmJ_{h\hat\rme_N}(\hat\rme_k)=0,\quad  1\le k<N.$$ 
A simple consequence of this fact is that, if we define $m^*:=\rmJ_{h\hat\rme_N}(\hat\rme_N)$, then, for any  translation invariant Gibbs measure $\nu$ of the spin \(O(N)\) model at inverse temperature \(\beta\) and magnetic field \(h\rme_N\), 
\begin{equation}\label{firstpoint}\nu(S_0\cdot s) = s^N m^*.\end{equation}
	The proof is a classical fact, which follows from differentiability of \(\psi\), see e.g. \cite[Proof of Theorem 2.5]{Biskup-2009} for the proof of a similar fact. A slightly more subtle consequence is that 
\begin{equation}\label{secondpoint}\nu(|M_n-\nu(S_0)|)\xrightarrow{n\to\infty} 0,\end{equation}
where \(M_{n} = \frac{1}{(2n+1)^d} \sum_{x\in \{-n,\dots, n\}^d} S_x \).
To prove this, we use the ergodic decomposition of \(\nu\): there exists a probability measure \(P_{\nu}\) on the space of ergodic Gibbs measures such that \(\nu(f) = \int dP_{\nu}(\eta) \eta(f)\) for every \(\nu\)-integrable \(f\) (see~\cite[Theorem 14.10]{Georgii-2011}). Moreover, by the ergodic theorem, for any \(\eta\) ergodic, \(\lim_{n\to \infty} \eta(|M_n-\eta(S_0)|) = 0\) (see~\cite[Appendix 14.A]{Georgii-2011}). Now, any \(\eta\) ergodic is in particular translation invariant, therefore, by \eqref{firstpoint}, \(\eta(S_0) = (0,\dots,0, m^*) = \nu(S_0) \). So, one has the pointwise convergence \(\lim_{n\to \infty} \eta(|M_n - \nu(S_0)|) = 0\). To conclude, by Vitali's dominated convergence theorem one has
	\begin{equation*}
		\lim_{n\to \infty}\nu(|M_n - \nu(S_0)|) = \lim_{n\to \infty}\int dP_{\nu}(\eta) \eta(|M_n - \nu(S_0)|) = 0.
	\end{equation*}
Given \eqref{secondpoint}, the proof of eqs.\eqref{eq:infrared_infin_vol_centered_spins}-\eqref{eq:infrared_magnet} follows \cite{Frohlich+Simon+Spencer-1976}. Let \(h\in D^*\), and let \(f:\Z^d\to \R^N\) be a finitely supported function. Then, set \(C_f= \sum_x f(x)\in \R^N\), and \(f_n(x) = f(x) - \frac{1}{(2n+1)^d} C_f \mathds{1}_{x\in \{-n,\dots, n\}^d}\). Since \(\sum_{x} f_n(x) = 0\), one can apply~\eqref{eq:infrared_fin_vol.0} and take \(L\to \infty\) to obtain
\begin{equation*}
	\lrangle{e^{\sum_{x} S_x\cdot f_n(x) }}_{h}\leq e^{\frac{1}{2\beta}(f_n,G f_n) }.
\end{equation*}Now, as in~\cite[Corollary 2.5]{Frohlich+Simon+Spencer-1976}, \((f_n,G f_n)\to (f,G f)\) as \(n\to \infty\) (using the decay of the Green function in \(d\geq 3\)). To conclude, observe that \(\sum_x S_x\cdot f_n(x) = \sum_x (S_x- M_n)\cdot f(x) \). Then, for every integer \(p\geq 1\),
\begin{multline*}
	\big|\lrangle{\big(\sum_{x} S_x\cdot f_n(x) \big)^p}_{h} - \lrangle{\big(\sum_{x} (S_x-\lrangle{S_0}_h)\cdot f(x) \big)^p}_{h}\big|= \\
	= \big|\sum_{k=1}^{p}\binom{p}{k} \lrangle{\big(\sum_{x} (S_x-\lrangle{S_0}_h)\cdot f(x) \big)^{p-k}\big( (\lrangle{S_0}_h-M_n) \cdot C_f \big)^k}_{h}\big|\\
	\leq \sum_{k=1}^{p}\binom{p}{k} (2c_f)^{p-k} (2|C_f|)^{k-1}|C_f|\lrangle{| \lrangle{S_0}_h-M_n|}_{h}
\end{multline*}where \(c_f = \sum_{x} |f(x)|\). Now, \eqref{secondpoint} implies that the last line converges to \(0\) as \(n\to \infty\); therefore, the convergence \(\lrangle{e^{\sum_{x} S_x\cdot f_n(x) }}_{h}\to \lrangle{e^{\sum_{x} (S_x-\lrangle{S_0}_h)\cdot f(x) }}_{h}\) as \(n\to \infty\) follows. The bound~\eqref{eq:infrared_magnet} follows from~\eqref{eq:infrared_infin_vol_centered_spins} as in~\cite{Frohlich+Simon+Spencer-1976} (choose $f(x)=\epsilon\delta_{x,0}\hat\rme_i)$, expand in $\epsilon$ at second order, 
and sum over $i=1,\ldots,N$, recalling that $\sum_{i=1}^N(S^i_0)^2=1$).

\section{Upper bounds on correlation functions}
\label{app:UB_corr_funct}

This appendix contains the proof of Lemma~\ref{lem:setof_pt_to_setof_pt_decay}. We first derive a bound on the two point function. The log can be removed when \(d\geq 4\) as is clear from the proof.

\begin{lemma}
	\label{lem:UB_twoPts_funct_spin}
	For any \(d\geq 3\), there exists \(c\geq 0\) such that for any distinct \(x,y\in\Zd\), and any \(k=1,\cdots,N\),
	\begin{equation}
	\label{eq:UB_twoPts_funct_spin}
		\lrangle{S_x^k; S_y^k}\leq \frac{c}{\beta}\frac{\log(n)}{n}
	\end{equation}where \(n=\normsup{x-y}\).
\end{lemma}
\begin{proof}
	By translation invariance, we can reduce to the case \(y=0\). Then, by symmetry, one can assume \(\normsup{x} = x_1\). Furthermore, using reflection positivity, one has
	\begin{equation*}
		|\lrangle{f_0f_x}| = |\lrangle{f_0 \Theta f_{x^{\perp}}}| \leq \lrangle{f_0 f_{x^{\parallel}}}^{1/2} \lrangle{f_{x^{\perp}} f_{x}}^{1/2},
	\end{equation*}where \(\Theta\) is the reflection through the hyperplane \(\{y:\ y_1= x_1/2 \}\), \(f_x = S_x^k\), \(1\leq k< N\) or \(f_x = S_x^N-\lrangle{S_0^N}\), and \(x=x^{\parallel} + x^{\perp}\), \(x^{\parallel} = (x\cdot \rme_1)\, \rme_1\). Since, by translation invariance, \(\lrangle{f_{x^{\perp}} f_x } = \lrangle{f_0f_{x^{\parallel}}}\), the general claim is implied by its restriction to \(x= n\rme_1\).
	
First treat \(k\neq N\). By symmetry, it is sufficient to prove the result for \(k=1\). Set \(g(x)= z\rme_1\mathds{1}_{x\in\{0,\ldots, n\rme_1\}}\). Use then~\eqref{eq:infrared_infin_vol_centered_spins} for this \(g\). Expand both sides of the inequality, simplify the constant term, divide by \(z^2\) and take \(z\searrow 0\) to obtain
	\begin{equation*}
		\sum_{k,l=0}^{n} \lrangle{S_{k\rme_1}^1S_{l\rme_1}^1} \leq \frac{1}{\beta} \sum_{k,l=0}^{n} G(k\rme_1, l\rme_1),
	\end{equation*}where \(G(k\rme_1, l\rme_1) \leq c/(1+|k-l|)^{d-2}\). We can then use the standard fact (following from reflection positivity, see~\cite{Seiler-1982} or~\cite{Lees+Taggi-2019}) that \(\lrangle{f_{0} f_{n\rme_1} }\) in monotonic in \(n\) (for the same \(f\)s as before) to obtain
	\begin{equation*}
		n^{2}\lrangle{S_0^1S_{n\rme_1}^1}\leq \frac{c}{\beta} n\sum_{k=0}^{n} (1+k)^{2-d} \leq \frac{c}{\beta} n\log(n).
	\end{equation*}The proof for \(k=N\) follows exactly the same path but the truncature becomes non-trivial.
	\end{proof}

\begin{proof}[Proof of Lemma~\ref{lem:setof_pt_to_setof_pt_decay}]
To simplify notations, we conduct the proof with infinite volume notation. What we are really doing is to work in finite volume with \(h>0\) and take limits afterwards. We will use several times that if one conditions on everything but the sign field of \(\phi\), the latter is distributed according to a ferromagnetic Ising model, denoted \(\sigma\) in what follows, and is therefore coupled to a Random Cluster (RC) model \cite{Fortuin+Kasteleyn-1972}. We denote \(\lrangle{\cdot}'\) the conditional measure of the sign field \(\sigma\), and \(\Phi'\) the associated Random Cluster measure. We first prove the intermediate bound: there exist \(C<\infty, c>0\) such that for any \(x,y\in\Zd\) and \(T>0\),
	\begin{equation}
	\label{eq:UB_phi_TwoPtsFunct}
		\lrangle{\phi_x \phi_y}\leq \frac{C\log(1+|x-y|)}{1+|x-y|} + Ce^{-c{\beta}}.
	\end{equation}Indeed, introducing the event \(B=\{|\theta_x|\leq \frac{\pi}{2} \}\cap\{|\theta_y|\leq \frac{\pi}{2}\}\cap\{\norm{u_x}\leq \frac{\sqrt{3}}{2}\}\cap\{\norm{u_y}\leq \frac{\sqrt{3}}{2}\}\), one has (using the fact that, on $B$, \(\frac{1}{\rho_x\rho_y}\leq 4\), and \(\frac{|\theta_x\theta_y|}{|\sin\theta_x\sin\theta_y|}\leq \pi^2/4\))
	\begin{align*}
		0\leq \lrangle{\phi_x \phi_y} &= \beta\lrangle{\mathds{1}_B|\theta_x \theta_y|\lrangle{\sigma_x\sigma_y}'} + \beta\lrangle{\mathds{1}_{B^c}|\theta_x \theta_y|\lrangle{\sigma_x\sigma_y}'}\\
		&\leq \pi^2\beta\lrangle{\mathds{1}_B\rho_x\rho_y|\sin\theta_x\sin\theta_y|\lrangle{\sigma_x\sigma_y}'} + \pi^2\beta\mu(B^c)\\
		&\leq \pi^2\beta\Big(\lrangle{S_x^{N-1} S_{y}^{N-1}} + \beta\mu(B^c)\Big)
	\end{align*}
	To conclude, apply Lemma~\ref{lem:UB_twoPts_funct_spin} and observe that
	\begin{equation*}
		\mu(B^c)\leq 2\mu(|\theta_0|>\pi/2) + 2\mu(\|u_0\|>\sqrt{3}/2)\leq Ce^{-c\beta}
	\end{equation*}by a union bound and an application of eqs.\eqref{eq:sign_Sk_non_neg}-\eqref{eq:sign_SN_non_neg}. Before we turn to the proof
of our main claim, let us also observe the following: define
	\begin{equation}\label{deftildeC}
		\tilde{C}_K = \sup_{0<T<T_0}\sup_{\substack{p,p',\calF\\ \normI{p}+ p'\leq K}} \lrangle{|\phi^p\calF|}
	\end{equation}where the second sup is over \(p:\Z^d\to \Z_+\), \(p'\in\Z_+\) and \(\calF\) monomial in \(\tilde{u}_x^k\) with degree at most \(2p'\). By Theorem~\ref{thm:moments_u_theta} and Hölder's inequality, \(\tilde{C}_K<\infty\). 
	
Let us now turn to the proof of \eqref{lemma:eq1}. 
Let \(A=\{x:\ p_x\text{ is odd}\}\) and \(B= \{x:\ q_x\text{ is odd}\}\). 
Note that, since \(\sum_x p_x\) and \(\sum_x q_x\) are odd, \(|A|\) and $|B|$ are odd as well. 
Using the RC representation of the Ising model associated with the $\sigma$ variables discussed above,
one has 
\begin{equation}
\label{boundRC}	|\lrangle{\calF \phi^p \phi^{q}}| \leq \lrangle{|\calF \phi^p \phi^{q}| \lrangle{\sigma_A\sigma_{B}}'} 
\end{equation}
Let now \(\calE_{D}\) be the percolation event ``each cluster contains an even (possibly \(0\)) number of sites of \(D\)''. Since \(|A|\) and $|B|$ are odd, \(\calE_{A\cup B}\) implies the existence of \(x\in A,y\in B\) such that \(x\) and \(y\) are in the same connected component, denoted \(x\leftrightarrow y\). One has
	\begin{equation*}
		\lrangle{\sigma_A\sigma_{B}}' = \Phi'(\calE_{A\cup B}) \\
		\leq \sum_{x\in A}\sum_{y\in B} \Phi'(x\leftrightarrow y) = \sum_{x\in A}\sum_{y\in B} \lrangle{\sigma_x\sigma_y}'.
	\end{equation*}
	Plugging this bound in \eqref{boundRC}, one obtains
	\begin{equation}
	\label{eq:proof:UB_corr_eq2}
		|\lrangle{\calF \phi^p \phi^{q}}| \leq \sum_{x\in A}\sum_{y\in B} \lrangle{|\calF \phi^{p-\delta_x} \phi^{q-\delta_y}| \phi_x\phi_y}.
	\end{equation} We then partition on whether \(|\calF \phi^{p-\delta_x} \phi^{q-\delta_y}|\leq \beta^{\epsilon}\) or not, and use Cauchy-Schwartz to obtain
	\begin{align}\label{prooflemmathird}
		\lrangle{|\calF \phi^{p-\delta_x} \phi^{q-\delta_y}| \phi_x\phi_y} &\leq \beta^{\epsilon}\lrangle{ \phi_x\phi_y} + \lrangle{|\calF \phi^{p} \phi^{q}|^2}^{1/2} \mu(|\calF \phi^{p-\delta_x} \phi^{q-\delta_y}| > \beta^{\epsilon})^{1/2}\\
		&\leq \beta^{\epsilon}\lrangle{ \phi_x\phi_y} + (\tilde{C}_{2K_0})^{1/2} \Big(\frac{\tilde{C}_{K_0\lceil 2\alpha/\epsilon\rceil}}{\beta^{2\alpha}}\Big)^{1/2}
	\end{align}where we used the definition \eqref{deftildeC}, we let $K_0=\normI{p}+\normI{q}+ p'$, and in the last line we used Markov's inequality. Applying the upper bound on correlation~\eqref{eq:UB_phi_TwoPtsFunct} in this last estimate and plugging the resulting bound in~\eqref{eq:proof:UB_corr_eq2} and in~\eqref{eq:proof:UB_corr_eq1}, 
we obtain \eqref{lemma:eq1}. For \eqref{lemma:eq2}, letting $C'=\bigcup_{(x,\rme): q_{x,\rme}'\neq 0 } \{x,x+\rme\}$, 
we expand $$(\nabla\phi)^{q'} = \sum_{C\subset C'}\sum_{\substack{q:C\to\Z_+ \\ \normI{q} \leq \normI{q'}}} c_{q'}(q)\phi^{q}$$ 
for appropriate coefficients $c_{q'}(q)$.
We can then write
	\begin{equation}
	\label{eq:proof:UB_corr_eq1}
		\lrangle{\calF \phi^p (\nabla\phi)^{q'}} =  \sum_{C\subset C'}\sum_{\substack{q:C\to\Z_+ \\ \normI{q} \leq \normI{q'}}} c_{q'}(q)\lrangle{\calF \phi^p \phi^{q}}.
	\end{equation}Now, by symmetry, \(\lrangle{\calF \phi^p \phi^{q}}\) is \(0\) if \(\sum_{x\in C} q_x\) is even. In particular, defining \(C_{q} = \{x\in C:\ q_x\text{ odd}\}\), \(|C_{q}|\) is odd. Then we bound $\lrangle{\calF \phi^p \phi^{q}}$ as in eqs.\eqref{boundRC} to \eqref{prooflemmathird} and, using \eqref{eq:UB_phi_TwoPtsFunct}, 
	we obtain \eqref{lemma:eq2}, as desired. 
\end{proof}

\section{Gaussian estimates}\label{app:B}

\subsection{Correlations and massive correlations}

Let us first state and prove the following standard bounds on the massive lattice Green's function. 
\begin{theorem}
	\label{thm:Gaussian_asymptotics}
	For $d\ge 3$, there exists \(c=c(d)\) such that for any \(m\geq 0\), \(x\in\Zd\), \(e,e'\in\rmB\),
	\begin{equation*}
		0\le G^m_{0x}\leq \frac{c}{(1+|x|)^{d-2}},\quad |\nabla_x^\rme G^m(0,\cdot)| \leq \frac{c}{(1+|x|)^{d-1}},\quad |\nabla_{0,x}^{\rme',\rme} G^m| \leq \frac{c}{(1+|x|)^{d}}.
	\end{equation*}
\end{theorem}
\begin{proof}
We use the following integral representation of \(G^m\) (see, for example,~\cite{Michta+Slade-2021}):
\begin{equation}
	\label{eq:Green_function_Bessel_rep}
	G^m(0,x) = \int_{0}^{\infty}dt e^{-(m^2+2d)t} \prod_{k=1}^d I_{x_k}(2t),
\end{equation}where \(I_{\nu}(t)\) is the modified Bessel function of the first kind. We shall use that \(I_{\nu}(t)\geq 0\), \(I_{-\nu}(t) = I_{\nu}(t)\), and \(I_{\nu}(t)\geq I_{\nu+1}(t)\) for \(\nu\in\Z_+, t\in\R_+\) \cite{Soni-1965}. The asymptotic \(I_{\nu}(t)\sim O((2\pi t)^{-1/2}e^{t})\) valid for large real argument makes the integral well defined at \(m=0\) for any \(d\geq 3\).

	The claim at \(m=0\) is classical and proven, for example, in~\cite{Lawler+Limic-2010}. The integral representation~\eqref{eq:Green_function_Bessel_rep} together with the monotonicity of \(I_{\nu}(t)\) in \(\nu\) implies \(G^m_{0x}\), \(|\nabla_x^\rme G^m(0,\cdot)|\), and \(|\nabla_{0,x}^{\rme',\rme}G^m|\) when \(\rme\cdot \rme'= 0\) are 
	monotone decreasing in $m$; therefore, in these cases, the bound for $m=0$ implies the one for $m\neq0$. 
By translation invariance, we are then left with bounding \(|\nabla_{0,x}^{\rme',\rme}G^m|\) for \(\rme= \rme'\); we can further restrict to \(\rme=\rme_1\) by symmetry. Consider first \(x_1=0\). Then, using~\eqref{eq:Green_function_Bessel_rep},
	\begin{equation*}
		|\nabla_{0,x}^{\rme_1,\rme_1} G^m| =  2\int_{0}^{\infty}dt e^{-(m^2+d)t} \big( I_{0}(t)-I_{1}(t) \big)\prod_{k=2}^{d} I_{x_k}(t),
	\end{equation*}which is again decreasing in \(m\). We then consider \(x_1\geq 1\) (which will cover all leftover cases by symmetry). W.l.o.g., we can suppose \(x_2,\cdots,x_{d}\geq 0\) and \(x_2=\max(x_2,\cdots,x_d)\). We first treat \(x_1\geq x_2\). Using the reflection positivity of the Gaussian measure $\Phi_m$, and denoting by $\Theta$
the reflection through the hyperplane \(\{y:\ y_2= \lfloor x_1/2\rfloor -y_1\}\) (which passes through sites of $\Lambda$ and is tilted by an angle $\pi/4$ w.r.t. the $i$-th coordinate axis, $i>2$)), we find: 
	\begin{multline*}
		|\nabla_{0,x}^{\rme_1,\rme_1} G^m| = |\Phi_m(\nabla_{0}^{\rme_1} \phi \nabla_{x}^{\rme_1} \phi)| = |\Phi_m(\nabla_{0}^{\rme_1} \phi \Theta \nabla_{\bar{x}}^{-\rme_2} \phi)|\leq\\
		\leq \Phi_m(\nabla_{0}^{\rme_1} \phi \Theta \nabla_{0}^{\rme_1}\phi_0)^{1/2} \Phi_m(\nabla_{\bar{x}}^{-\rme_2} \phi \Theta \nabla_{\bar{x}}^{-\rme_2} \phi)^{1/2} = \Phi_m(\nabla_{0}^{\rme_1} \phi \nabla_{\tilde{x}}^{-\rme_2}\phi)^{1/2} \Phi_m(\nabla_{\bar{x}}^{-\rme_2} \phi \nabla_{x}^{\rme_1} \phi)^{1/2}
	\end{multline*}where \(\bar{x} = (\lfloor x_1/2\rfloor-x_2, \lfloor x_1/2\rfloor- x_1,x_3,\cdots,x_d)\) and \(\tilde{x} = (\lfloor x_1/2\rfloor,\lfloor x_1/2\rfloor,0,\cdots, 0 )\). Now, as \(e\cdot \rme_2 =0\), \( \norm{\tilde{x}}\geq cx_1\geq c d^{-1/2}\norm{x}\), and \( \norm{\bar{x}- x} \geq c' x_1 \geq c' d^{-1/2}\norm{x}\), we can apply the previously obtained bounds to get the result. The case \(x_2>x_1\) is treated similarly using a reflection through \(\{y:\ y_2= \lfloor x_2/2\rfloor -y_1\}\).
\end{proof}

We next gather the following estimates on the sum of $G^m$ and of their derivatives, as well as on the difference between $G^m$ and  $G^0$, which are used systematically in the proof of our main result. 

\begin{theorem}
	\label{thm:Gaussian_estimates}
	There exists \(c<\infty\) such that for any \(0<m\le 1\),
	\begin{enumerate}
		\item \begin{equation}
		\label{eq:Gaussian_est:sum_of_G}
			\sum_{x\in\Zd} G^m(0,x) = \frac{1}{m^2}
		\end{equation}
		\item For any \(x\) \begin{equation}
		\label{eq:Gaussian_est:mass_removal_G}
			|G^{m}(0,x) - G(0,x)|\leq c m.
		\end{equation}
		\item for \(\rme \in\rmB\),
		\begin{equation}
		\label{eq:Gaussian_est:sum_of_gradG}
			\sum_{x\in \Zd} |\Phi_m\big( \phi_0(\phi_x-\phi_{x+\rme})\big)|\leq c m^{-1}.
		\end{equation}
		\item for \(\rme,\rme' \in\rmB\),
		\begin{equation}
		\label{eq:Gaussian_est:sum_of_gradgradG_general}
			\sum_{x\in \Zd} |\Phi_m\big((\phi_0-\phi_{\rme})(\phi_x-\phi_{x+\rme'})\big) |\leq c|\log(m)|.
		\end{equation}
		\item For \(0<\alpha<1\), there exists \(C<\infty\) such that for any \(y\in\bbZ^d\),
		\begin{gather}
		\begin{split}
		\label{eq:Gaussian_est:sum_of_G_gradG_with_power}
			\sum_{x\in\Zd} G^m(0,x) (1+|x-y|)^{-\alpha} \leq Cm^{-2+\alpha},\\
			\sum_{x\in\Zd} |\Phi_m\big( \phi_0(\phi_x-\phi_{x+\rme})\big)| (1+|x-y|)^{-\alpha} \leq Cm^{-1+\alpha}.
		\end{split}
		\end{gather}
		\item For every \(k\geq 0\) there exists \(C<\infty\) such that for any \(y\in\bbZ^d\),
		\begin{equation}
		\label{eq:Gaussian_est:sum_of_gradgradG_with_power}
			\sum_{x\in\Zd} |\Phi_m\big( (\phi_0-\phi_{\rme'})(\phi_x-\phi_{x+\rme})\big)| \frac{\log^k(1+|x-y|)}{1+|x-y|}\leq C\frac{\log^{k+1}(1+|y|)}{(1+|y|)}.
		\end{equation}
	\end{enumerate}
\end{theorem}

\begin{proof}
	All estimates are proved in~\cite{Bricmont+Fontaine+Lebowitz+Spencer-1980}, except for the last one. More precisely, 
	for \eqref{eq:Gaussian_est:sum_of_G}, 
see \cite[Prop.A1(a)]{Bricmont+Fontaine+Lebowitz+Spencer-1980}; for \eqref{eq:Gaussian_est:mass_removal_G}, 
see \cite[Prop.A5]{Bricmont+Fontaine+Lebowitz+Spencer-1980}; for \eqref{eq:Gaussian_est:sum_of_gradG},
see \cite[Prop.A1(b)]{Bricmont+Fontaine+Lebowitz+Spencer-1980}; for \eqref{eq:Gaussian_est:sum_of_gradgradG_general}, 
see \cite[Prop.A1(d)]{Bricmont+Fontaine+Lebowitz+Spencer-1980}; for \eqref{eq:Gaussian_est:sum_of_G_gradG_with_power}
see \cite[Prop.A1(e)]{Bricmont+Fontaine+Lebowitz+Spencer-1980}. We are left with proving \eqref{eq:Gaussian_est:sum_of_gradgradG_with_power}: 
using Theorem~\ref{thm:Gaussian_asymptotics}, we have the upper bound
	\begin{equation*}
		\sum_{x\in\Zd} \frac{c}{(1+ |x|)^d} \frac{\log^k(1+|x-y|)}{1+|x-y|} = \sum_{x:|x|\leq |y|/2} \cdot + \sum_{x:|x-y|< |y|/2}\cdot  + \sum_{x:|x|> |y|/2, |x-y|\geq |y|/2}\cdot.
	\end{equation*}
	In the first sum, one has \(|y|/2\le |x-y|\le 3|y|/2\), which allows us to bound from above the first term by \( C'\frac{\log^k(1+|y|)}{1+|y|}  \sum_{|x|\leq |y|/2} (1+|x|)^{-d}\le C
	\frac{\log^k(1+|y|)}{1+|y|} \log(1+|y|)\).  In the same fashion, in the second sum, \(|x|> |y|/2\) so the second term is  bounded from above by \(C'\frac{\log^k(1+|y|)}{(1+|y|)^d}
	\sum_{|x-y|<|y|/2}(1+|x-y|)^{-1}\le C\frac{\log^k(1+|y|)}{1+|y|}\). In the last case, using H\"older's inequality, we see that the third sum is  bounded from above by
	$$C'\Big(\sum_{|x|>|y|/2}\frac1{(1+|x|)^{d+1}}\Big)^{\frac{d}{d+1}}\Big(\sum_{|x-y|\ge |y|/2}\frac{\log^{k(d+1)}(1+|x-y|)}{(1+|x-y|)^{d+1}}\Big)^{\frac{1}{d+1}}\le C\frac{\log^k(1+|y|)}{1+|y|}.$$  This proves the claim (\(C\) depends on \(k,d\)).
\end{proof}


In order to compare massive and massless propagators, we use Theorem \ref{thm:Gaussian_asymptotics} in combination with eq.\eqref{eq:Gaussian_est:mass_removal_G} and the following Lemma.
\begin{lemma}
	\label{lem:decay_and_mass_removal_implies_massDecay}
	Suppose \(f_m(x)\) is such that
	\begin{itemize}
		\item \(\sup_{m\geq 0}|f_m(x)|\leq C(1+|x|)^{-\gamma}\),
		\item \(|f_m(x)-f_0(x)| \leq Cm\),
	\end{itemize}for some \(C,\gamma>0\). {Then, for any $0<\gamma'<\gamma$, there exists \(C'>0\) such that
	\begin{equation}\label{concl}
		|f_m(x)-f_0(x)|\leq C' m^{1-\gamma'/\gamma}(1+|x|)^{-\gamma'}.
	\end{equation}}
\end{lemma}

\begin{proof} One has \(|f_m(x)-f_0(x)|\leq |f_m(x)-f_0(x)|^{1-\gamma'/\gamma}(2\sup_{m\geq 0}|f_m(x)|)^{\gamma'/\gamma}$. From this, the two assumptions readily imply 
\eqref{concl}.
\end{proof}

We conclude this subsection with a simple lemma that we used in many instances.

\begin{lemma}
	\label{app:lem:sum_Zd_a_b_to_log}
	Let \(a,b> 0\) be such that \(a+b >d\). Then, there exists \(C=C(a,b)\) such that for any \(u\in \Zd\),
	\begin{equation*}
		\sum_{x\in \Zd} (1+|x|)^{-a} (1+|x-u|)^{-b} \leq C \frac{\log(1+|u|)}{(1+|u|)^{a+b-{d}}}.
	\end{equation*}
\end{lemma}

\begin{proof}
	Separate the sum over \(x\) into \(\sum_{|x|\leq |u|/2} + \sum_{|x-u|< |u|/2} + \sum_{\substack{|x|> |u|/2\\ |x-u| \ge |u|/2}}\). In the first sum, one has \(|x-u|\geq |u|/2\), which allows us to bound from above the first term by \( C \frac{\log(1+|u|)}{(1+|u|)^{a+b-d}}\) (the log only appears if \(a=d\)). In the same fashion, in the second sum, \(|x|> |u|/2\) so the second term is  bounded from above by \(C \frac{\log(1+|u|)}{(1+|u|)^{a-d+b}}\) (the log only appears if \(b=d\)). In the last case, one has \(|x|\leq |x-u|+|u|\leq 3|x-u|\). So, the last sum is  bounded from above by \(C\sum_{|x|> |u|/2} (1+|x|)^{-a-b} \leq C(1+|u|)^{d-a-b}\).
\end{proof}

\subsection{Gaussian connected correlations}

Let \(a_1,\cdots, a_n\in \Zd\) and denote by \(A=\{a_1,\cdots,a_n\}\) the corresponding multi-set, i.e., the collection of the elements $a_1,\ldots,a_n$ counted with their multiplicities, in the case that a given site of $\mathbb Z^d$ appears more than once in the list. One has the Wick formula
\begin{equation}
\label{eq:Wick_rule}
	\Phi(\phi_A) = \sum_{\pi \in \pairings(A)} \prod_{\{a_1,a_2\}\in\pi}\Phi(\phi_{a_1}\phi_{a_2}), 
\end{equation}where $\phi_A=\prod_{a\in A}\phi_a$, and
 \(\pairings(A)\) is the set of partitions of \(A\) into two-elements classes. 
 Moreover, from \eqref{Urs2} one also has the corresponding formula for connected correlations:
   let \(A_i=\{a_1^i,\cdots, a_{k_i}^i\}\), \(i=1,\cdots, n\), be multi-sets of sites. Then,
\begin{equation}
\label{eq:Wick_rule_connected}
	\Phi(\phi_{A_1}; \cdots; \phi_{A_n}) = \sum_{\pi \in \pairings^c(A_1,\cdots,A_n)} \prod_{\{a_1,a_2\}\in\pi}\Phi(\phi_{a_1}\phi_{a_2}),
\end{equation}where \(\pairings^c(A_1,\cdots,A_n)\) is the subset of \(\pairings(A_1\cup \cdots \cup A_n)\) consisting of the pairings $\pi$ such that 
the graph with vertex set \(\{A_1,\cdots,A_n\}\) and edge set \(F_{\pi}=\{\{A_i,A_j\}: \exists\ a\in A_i, b\in A_j \ \text{such that}\ \{a,b\}\in \pi\}\) is connected. The same formulas extend (by linearity) to the case \(A_i = A_i^v\sqcup A_i^e\) with \(A_i^v\) a multi-set of sites and \(A_i^e\) a multi-set of edges (elements of \(\bbE_d\)) and \(\phi_{A_i}\) is replaced by \(\phi_{A_i^v }(\nabla \phi)_{A_{i}^e}\).

\begin{lemma}
	\label{app:lem:sum_of_connected_corr}
	Let \(K, n\geq 0\) be integers. Let \(\epsilon>0\). Then, there exists \(C=C(n,K,\epsilon)\) such that for any \(p,p':\Zd\to \Z_+\) odd, any \(q:\bbE_d\to \Z_+\) odd, and \(r_1,\cdots,r_n>0\) integers with \(\normI{p}+\normI{p'}+\normI{q}+\sum_{i=1}^n 2r_i\leq K\),
	\begin{equation}\label{eq:71}
		\sum_{\substack{x_1,\cdots, x_n\in \Zd\\ e_1,\cdots, e_n\in \rmB_+}}\big|\Phi\big(\phi^p \phi^{p'}; (\nabla_{x_1}^{e_1}\phi )^{2r_1};\cdots; (\nabla_{x_n}^{e_n}\phi )^{2r_n} \big)\big| \leq C\sum_{x\in \supp_p}\sum_{y\in\supp_{p'}} \frac{1}{(1+|x-y|)^{d-2-\epsilon}},
	\end{equation}and
	\begin{equation}\label{eq:72}
		\sum_{\substack{x_1,\cdots, x_n\in \Zd\\ e_1,\cdots, e_n\in \rmB_+}}\big|\Phi\big(\phi^p (\nabla\phi)^{q}; (\nabla_{x_1}^{e_1}\phi )^{2r_1};\cdots; (\nabla_{x_1}^{e_n}\phi )^{2r_n} \big)\big| \leq C\sum_{x\in \supp_p}\sum_{y\in\supp_{q}} \frac{1}{(1+|x-y|)^{d-1-\epsilon}}.
	\end{equation}
\end{lemma}
\begin{proof} For $n=0$ the summands in the left sides should be interpreted as $|\Phi(\phi^p\phi^{p'})|$ and $|\Phi(\phi^p(\nabla\phi)^q)|$, respectively. In this case, 
the claim readily follows from~\ref{thm:Gaussian_asymptotics} and~\eqref{eq:Wick_rule}. We therefore suppose \(n\geq 1\). As the two claims are proved in exactly the same fashion, we prove only the first one and leave the adaptations to get the second to the reader. We will use the identity~\eqref{eq:Wick_rule_connected} together with Lemma~\ref{app:lem:sum_Zd_a_b_to_log}.
	
	It will useful to use a graphical language: introduce a set of sites \(\{0, 0', 1,\cdots, n\}\) corresponding to \(\phi^p,\phi^{p'},(\nabla_{x_1}^{e_1}\phi)^{2r_1}, \cdots, (\nabla_{x_n}^{e_n}\phi)^{2r_n} \). Edges will be given by the pairings in~\eqref{eq:Wick_rule_connected}. See Figure~\ref{fig:graphical_rep_connected_corr} for an illustration with \(n=4\).
	\begin{figure}
	\centering
	\includegraphics[scale=0.5]{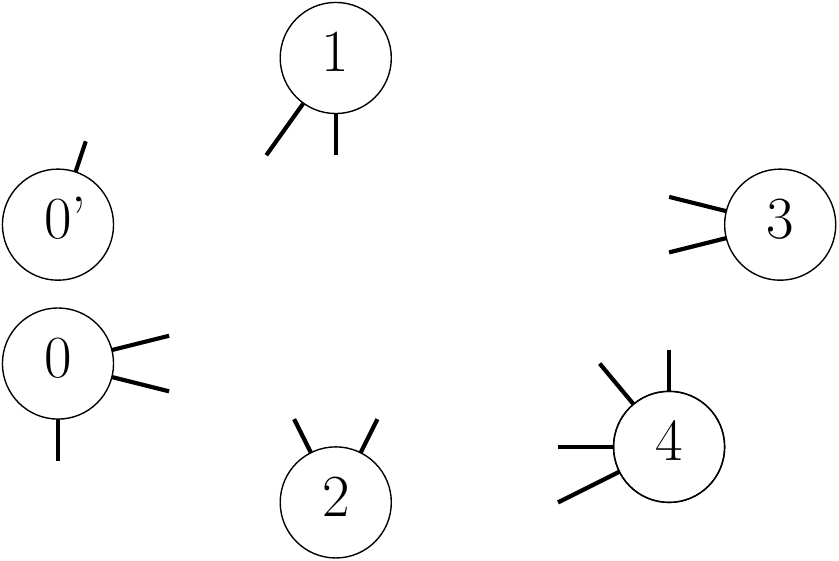}
	\hspace{2cm}
	\includegraphics[scale=0.5]{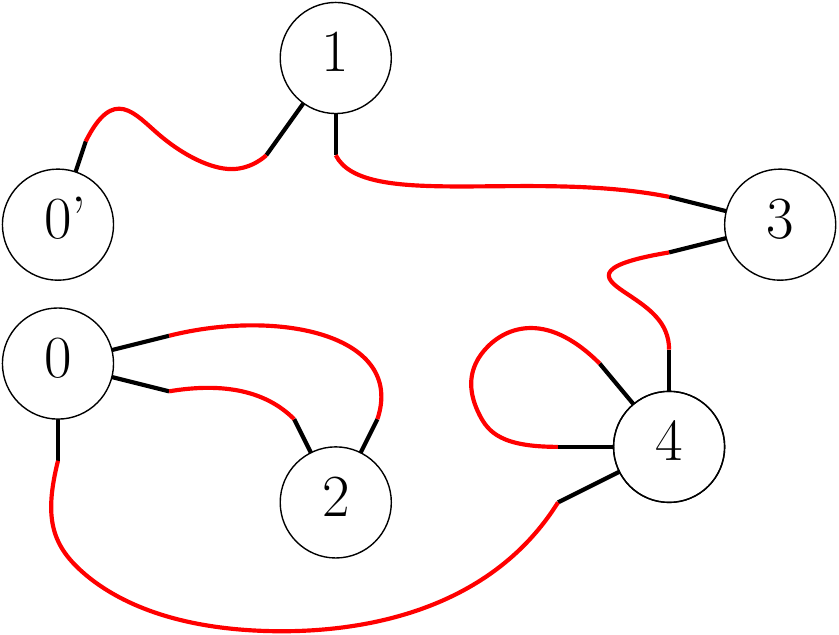}
	\caption{\(n=4\), \(r_1=r_2=r_3=1\), \(r_4=2\). Left: the ``vertices'' with the outgoing half-edges. Right: a possible pairing of the elements with \(0,0'\) not merged.}
	\label{fig:graphical_rep_connected_corr}
	\end{figure}
	
	We consider two cases: first, if there is a pairing \(\{x,y\}\) in \(\pi\) such that \(x\in\supp_p,y\in\supp_{p'}\), we merge the two vertices \(0,0'\) into one, to be denoted $00'$ (see Figure~\ref{fig:graphical_rep_connected_corr2}). In this case, since all the terms in the connected correlation have even degree, one can extract (for a given pairing \(\pi\)) two sites \(u,z\in \supp_p\cup\supp_{p'}\), and a path \(\gamma: 00'\to 00'\) passing through every site in \(\{1,\cdots, n\}\), of length \(|\gamma|\leq K+1\) (see Figure~\ref{fig:graphical_rep_connected_corr2}) 	
	\begin{figure}
		\centering
		\includegraphics[scale=0.5]{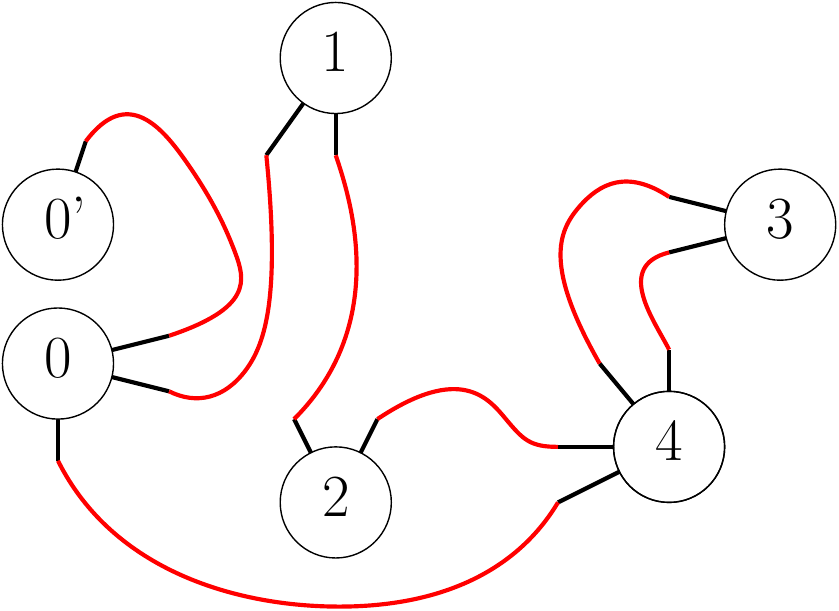}
		\hspace{2cm}
		\includegraphics[scale=0.5]{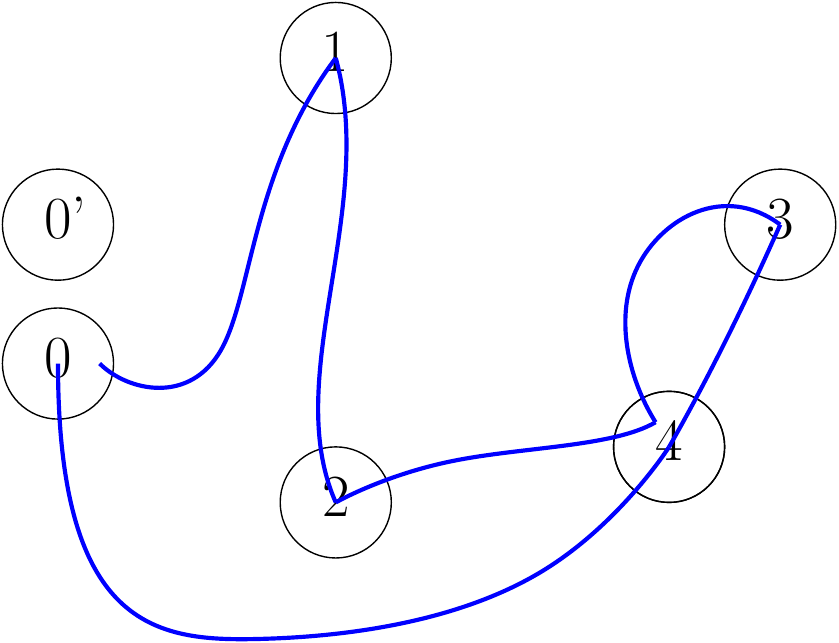}
		\caption{\(n=4\), \(r_1=r_2=r_3=1\), \(r_4=2\). Left: a possible pairing of the elements with \(0,0'\) merged. Right: a path extraction.}
		\label{fig:graphical_rep_connected_corr2}
	\end{figure}
such that the contribution to the total weight of \(\pi\) coming from the pair \(x,y\) is  \(\Phi(\phi_x\phi_y)\) while the contribution coming from \(\gamma\) is
	\begin{equation*}
		\Phi\big(\phi_u\nabla_{x_{\gamma_1}}^{e_{\gamma_1}}\phi \big)\Big(\prod_{k=1}^{M-2} \Phi\big(\nabla_{x_{\gamma_k}}^{e_{\gamma_k}}\phi\nabla_{x_{\gamma_{k+1}}}^{e_{\gamma_{k+1}}}\phi \big) \Big)\Phi\big(\phi_z \nabla_{x_{\gamma_{M-1}}}^{e_{\gamma_{M-1}}}\phi \big)	\end{equation*}
(here $\gamma$ is a tuple $\gamma=(\gamma_0,\gamma_1,\ldots,\gamma_M)$, with  
$\gamma_0=\gamma_M=00'$, $\gamma_i\in\{1,\ldots,n\}$ for $i=1,\ldots,M-1$, such that the multi-set $\{\gamma_1,\ldots,\gamma_{M-1}\}$ contains $\{1,\ldots,n\}$). 
Moreover, the contribution of the other pairs is bounded from above by \(c^K\) for some universal \(c<\infty\). Therefore,
in the case \(0,0'\) are merged, we have the following upper bound on our target:
	\begin{equation*}
		c^K\sum_{\substack{x\in \supp_p\\ y\in\supp_{p'}}} \Phi(\phi_x\phi_y) \sum_{u,z\in\supp_{p+p'}}\sum_{\gamma:u\to z} \sum_{\substack{x_1,\cdots,x_n\\ e_1,\cdots,e_n}} \Big|\Phi(\phi_u \nabla_{x_{\gamma_1}}^{e_{\gamma_1}}\phi) \Phi(\phi_z \nabla_{x_{\gamma_{M-1}}}^{e_{\gamma_{M-1}}}\phi) \prod_{k=1}^{M-2} \Phi(\nabla_{x_{\gamma_{k}}}^{e_{\gamma_{k}}} \phi \nabla_{x_{\gamma_{k+1}}}^{e_{\gamma_{k+1}}} \phi)\Big|,
	\end{equation*}where the sum is over \(\gamma\) as described previously. The contribution of a fixed \(\gamma\) is then bounded from above by
	\begin{multline*}
		\sum_{\substack{x_1,\cdots,x_{M-1}\\ e_1,\cdots,e_{M-1}}} \Big|\Phi(\phi_u \nabla_{x_{1}}^{e_{1}}\phi) \Phi(\phi_z \nabla_{x_{M-1}}^{e_{M-1}}\phi) \prod_{k=1}^{M-2} \Phi(\nabla_{x_{k}}^{e_{k}} \phi \nabla_{x_{k+1}}^{e_{k+1}} \phi)\Big|\leq \\
		\leq C^{M} \sum_{x_1,\cdots,x_{M-1}} (|u-x_1|+1)^{1-d}(|z-x_{M-1}|+1)^{1-d} \prod_{k=1}^{M-2} (|x_k-x_{k+1}|+1)^{-d}.
	\end{multline*}Repeated use of Lemma~\ref{app:lem:sum_Zd_a_b_to_log} gives the summability of the infinite sum, implying the claim (in the case that \(0,0'\) are merged), since the combinatorial factor coming from the sum over \(\gamma\) depends only on \(n,K\).
	
	We next turn to the case where there is no pairing in \(\pi\) linking \(0\) to \(0'\) (such as in Figure~\ref{fig:graphical_rep_connected_corr}). In this case one can find \(x\in\supp_p\), \(y\in\supp_{p'}\) and \(\gamma:0\to 0\) passing at least once through every sites of \(\{1,\cdots,n\}\) and with \(|\gamma|\leq K+1\). 
The contribution to the total weight of $\pi$ from such a path is 
	\begin{equation*}
		\Phi\big(\phi_x\nabla_{x_{\gamma_1}}^{e_{\gamma_1}}\phi \big)\Phi\big( \phi_y\nabla_{x_{\gamma_{M-1}}}^{e_{\gamma_{M-1}}}\phi \big)\prod_{k=1}^{M-2} \Phi\big(\nabla_{x_{\gamma_k}}^{e_{\gamma_k}}\phi\nabla_{x_{\gamma_{k+1}}}^{e_{\gamma_{k+1}}}\phi \big).
	\end{equation*}up to additional factors from the other pairs, bounded from above by \(c^K\) for some universal \(c<\infty\). 
	Proceeding as before, we obtain the following upper bound on the target quantity:
	\begin{equation*}
		c^K\sum_{\substack{x\in \supp_p\\ y\in\supp_{p'}}} \sum_{\gamma:x\to y} \sum_{\substack{x_1,\cdots,x_n\\ e_1,\cdots,e_n}} \Big|\Phi(\phi_x \nabla_{x_{\gamma_1}}^{e_{\gamma_1}}\phi) \Phi(\phi_y \nabla_{x_{\gamma_{M-1}}}^{e_{\gamma_{M-1}}}\phi) \prod_{k=1}^{M-2} \Phi(\nabla_{x_{\gamma_{k}}}^{e_{\gamma_{k}}} \phi \nabla_{x_{\gamma_{k+1}}}^{e_{\gamma_{k+1}}} \phi)\Big|.
	\end{equation*}The contribution of a fixed path \(\gamma\) is then bounded from above as before by \(C^{M}\) times
	\begin{equation*}
		\sum_{x_1,\cdots,x_{M-1}} (|x-x_1|+1)^{1-d}(|y-x_{M-1}|+1)^{1-d} \prod_{k=1}^{M-2} (|x_k-x_{k+1}|+1)^{-d}.
	\end{equation*}
Once again, repeated use of Lemma~\ref{app:lem:sum_Zd_a_b_to_log} gives the summability of the infinite sum, implying the claim in the case that
there is no pairing in \(\pi\) linking \(0\) to \(0'\). This concludes the proof. 
\end{proof}

\bibliographystyle{plain}
\bibliography{ON_exp_bib}

\end{document}